\UseRawInputEncoding

\documentclass[letterpaper,11pt]{article}
\usepackage[margin=1in]{geometry}  

\usepackage{listings}
\usepackage{amsmath}
\usepackage{amsthm}
\usepackage{tikz}
\usepackage{caption}
\usepackage{array}
\usepackage{mdwmath}
\usepackage{multirow}
\usepackage{mdwtab}
\usepackage{eqparbox}
\usepackage{multicol}
\usepackage{amsfonts}
\usepackage{tikz}
\usepackage{multirow,bigstrut,threeparttable}
\usepackage{amsthm}
\usepackage{array}
\usepackage{bbm}
\usepackage{subfigure}
\usepackage{epstopdf}
\usepackage{mdwmath}
\usepackage{mdwtab}
\usepackage{eqparbox}
\usepackage{tikz}
\usepackage{latexsym}
\usepackage{cite}
\usepackage{amssymb}
\usepackage{bm}
\usepackage{amssymb}
\usepackage{graphicx}
\usepackage{mathrsfs}
\usepackage{epsfig}
\usepackage{psfrag}
\usepackage{setspace}
\usepackage[
CJKbookmarks=true,
bookmarksnumbered=true,
bookmarksopen=true,
colorlinks=true,
citecolor=red,
linkcolor=blue,
anchorcolor=red,
urlcolor=blue
]{hyperref}
\usepackage{algorithm}
\usepackage{algpseudocode}
\usepackage{stfloats}

\theoremstyle{plain}
\newtheorem{theorem}{Theorem}

\newtheorem{lemma}{Lemma}

\newtheorem{corollary}{Corollary}

\newtheorem{question}{Question}
\newtheorem*{question*}{Question}
\newtheorem{assumption}{Assumption}
\newtheorem{proposition}{Proposition}

\def \bP {\mathbb{P}}
\def \bE {\mathbb{E}}

\def \bR {\mathbb{R}}

\newcommand{\calA}{{\mathcal{A}}}

\newcommand{\calN}{{\mathcal{N}}}

\newcommand{\calU}{{\mathcal{U}}}
\newcommand{\calV}{{\mathcal{V}}}

\newcommand{\calX}{{\mathcal{X}}}

\newcommand{\stepa}[1]{\overset{\rm (a)}{#1}}
\newcommand{\stepb}[1]{\overset{\rm (b)}{#1}}
\newcommand{\stepc}[1]{\overset{\rm (c)}{#1}}
\newcommand{\stepd}[1]{\overset{\rm (d)}{#1}}

\begin{document}

\title{Geometric Lower Bounds for Distributed Parameter Estimation under Communication Constraints}

\author{Yanjun Han, Ayfer \"{O}zg\"{u}r, Tsachy Weissman\thanks{Y. Han, A. \"{O}zg\"{u}r, and T. Weissman are with the Department of Electrical Engineering, Stanford University, email: \url{{yjhan,aozgur,tsachy}@stanford.edu}}}

%


%
\maketitle

\begin{abstract}
We consider parameter estimation in distributed networks, where each sensor in the network observes an independent sample from an underlying distribution and has $k$ bits to communicate its sample to a centralized processor which computes an estimate of a desired parameter. We develop lower bounds for the minimax risk of estimating the underlying parameter for a large class of losses and distributions. Our results show that under mild regularity conditions, the communication constraint reduces the effective sample size by a factor of $d$ when $k$ is small, where $d$ is the dimension of the estimated parameter. Furthermore, this penalty reduces at most exponentially with increasing $k$, which is the case for some models, e.g., estimating high-dimensional distributions. For other models however, we show that the sample size reduction is re-mediated only linearly with increasing $k$, e.g. when some sub-Gaussian structure is available. We apply our results to the distributed setting with product Bernoulli model, multinomial model, Gaussian location models, and logistic regression which recover or strengthen existing results.

Our approach significantly deviates from existing approaches for developing information-theoretic lower bounds for communication-efficient estimation. We circumvent the need for strong data processing inequalities used in prior work and develop a geometric approach which builds on a new representation of the communication constraint. This approach allows us to strengthen and generalize existing results with simpler and more transparent proofs.
\end{abstract}

\tableofcontents


\section{Introduction}
Statistical estimation in distributed settings has gained increasing
popularity motivated by the fact that modern data sets are often distributed across multiple machines and processors,
and bandwidth and energy limitations in networks and within
multiprocessor systems often impose significant bottlenecks
on the performance of algorithms. There are also an increasing
number of applications in which data is generated in a
distributed manner and it (or features of it) are communicated
over bandwidth-limited links to central processors \cite{boyd2011distributed,balcan2012distributed,daume2012protocols,daume2012efficient,dekel2012optimal}. A notable example is the federated learning \cite{mcmahan2017communication}, where multiple entities collaborate in solving a machine learning problem under the coordination of a central server, and communication could be a primary bottleneck since wireless links and other end-user internet connections typically operate at low rates and can be potentially expensive and unreliable; see \cite{kairouz2019advances} for an overview. 

In this paper, we consider general distributed statistical estimation problems under communication constraints, and focus on the impact of a finite-communication budget per sample on the final estimation accuracy. More formally, consider the following parameter estimation problem
\begin{align*}
	X_1, X_2, \cdots, X_n \overset{i.i.d.}{\sim} P_\theta
\end{align*}
where we would like to estimate $\theta\in\Theta\subseteq \bR^d$ under some general loss function $L$ such as the $\ell_1$ or $\ell_2^2$ loss. Unlike the traditional setting where $X_1,\cdots,X_n$ are directly available to the estimator, we consider a distributed setting where each observation $X_i$ is available at a different sensor and has to be communicated to a central estimator by using a communication budget of $k$ bits. We consider a general interactive communication model known as the blackboard communication protocol $\Pi_{\mathsf{BB}}$ \cite{kushilevitz1997communication}: all sensors communicate via a publicly shown blackboard while the total number of bits each sensor can write in the final transcript $Y$ is limited by $k$. Note that when one sensor writes a message (bit) on the blackboard, all other sensors can see the content of the message. We assume that public randomness is available in the blackboard communication protocol. The main motivation for considering a blackboard communication protocol is that it models arbitrary interaction between the nodes. Impossibility results proven under this assumption provide insights about whether communication protocols that make (better) use  of interaction can potentially lead to better performance than those achieved by simple schemes, e.g. simultaneous protocols. We note that impossibility results are strongest when proven under this most flexible communication model. We also consider a weaker family of protocols called the simultaneous message passing protocols (denoted by $\Pi_{\mathsf{SMP}}$), where each sensor independently sends $k$ bits to the centralized processor. 

Under both models, the central sensor needs to produce an estimate $\widehat{\theta}$ of the underlying parameter $\theta$ from the the $k$-bit observations $Y^n$ it collects at the end of the communication. Our goal is to jointly design a communication protocol in $\Pi$ (which is either $\Pi_{\mathsf{BB}}$ or $\Pi_{\mathsf{SMP}}$) and the estimator $\widehat{\theta}(\cdot)$ so as to minimize the worst case risk, i.e. to characterize the following \emph{distributed minimax risk}
$$
R^\star(n,k,\Theta,\Pi) \triangleq \inf_{\Pi}\inf_{\widehat{\theta}}\sup_{\theta\in\Theta} \bE_{\theta}[L(\theta,\widehat{\theta})].
$$
In this paper, lower bounds of the distributed minimax risk will typically be shown under the stronger blackboard communication protocol, while upper bounds of the same order will be attainable under the weaker simultaneous message passing protocol. 

The main contributions of our paper are as follows:
\begin{enumerate}
	\item For a large class of statistical models, we develop a novel geometric approach that builds on a new representation of the communication constraint to establish information-theoretic lower bounds for distributed parameter estimation problems. Our approach circumvents the need for strong data processing inequalities, and relate the experimental design problem directly to an explicit optimization problem in high-dimensional geometry.
	\item Based on our new approach, we show that the communication constraint reduces the effective sample size from $n$ to $n/d$ for $k=1$ under mild regularity conditions, where $d$ is the dimension of the parameter to be estimated. Moreover, for general communication budget $k$, our new approach enables us to show that the penalty is at most exponential in $k$. 
	\item Our new approach also reveals that the tight dependence of the distributed minimax risk on $k$ is determined by different geometric inequalities in different statistical models. Our result recovers the linear dependence on $k$ when some sub-Gaussian structure is available, e.g., in the Gaussian location model. However, in models with heavier tails such as the high-dimensional distribution estimation model, we show that the exponential dependence on $k$ becomes tight. 
\end{enumerate}


\subsection{Related Work}
Distributed parameter estimation and learning, under communication or privacy constraints, have been considered in many recent works. Early work \cite{zhang2013information,shamir2014fundamental,garg2014communication,braverman2016communication,xu2017information} established strong data-processing inequalities to prove tight lower bounds of distributed minimax risk under communication constraints. In particular, they showed that for communication-constrained Gaussian mean estimation, the distributed minimax risk depends linearly on $k$ under the blackboard communication protocol. Similar approaches were also used to obtain minimax risks under privacy constraints \cite{duchi2013local,kairouz2016discrete,duchi2018minimax}. 

Due to the nature of strong data-processing inequalities, the above procedure typically leads to lower bounds linear in $k$ (or the squared privacy parameter). However, for certain statistical estimation problems such as the high-dimensional distribution estimation, this dependence may not be tight. For example, a tight exponential dependence on the privacy parameter was established in \cite{ye2018optimal} for discrete distribution estimation, and \cite{diakonikolas2017communication} established the tight total communication budget for this problem under the communication constraint\footnote{However, no full version of \cite{diakonikolas2017communication} with complete proofs is available online at the time of writing.}. A complete characterization of the distributed minimax risk for discrete distributione estimation and general $(n,k)$ was obtained in \cite{han2018distributed} under the simultaneous message passing protocol, with a tight exponential dependence on $k$. Generalization of \cite{han2018distributed} to general statistical models with varying dependence on $k$ was studied in an earlier version of this work \cite{han2018geometric}, but both papers require the usage of the simultaneous message passing procotol (and there was a technical mistake in handling the blackboard communication protocol). All the above work modeled the communication/privacy constraint directly and did not use the strong data-processing inequality; along this direction, a flourishing line of recent research has studied different distributed estimation \cite{acharya2019hadamard,acharya2019communication,acharya2020inference,acharya2020inference2,chen2020} and testing \cite{acharya2020inference,acharya2020inference2,acharya2019inference,acharya2020domain} problems for the discrete distribution model under the simultaneous message passing protocol. A similar Gaussian identity testing problem was also studied in \cite{acharya2020distributed}. 

However, although it is relatively easy to extend the strong data-processing inequality based approach to blackboard communication protocols, it is more difficult to extend the approach based on direct modeling  to interactive protocols $\Pi_{\mathsf{BB}}$. We review some recent work which dealt with interactive communication protocols. Duchi and Rogers \cite{duchi2019lower} established lower bounds for general interactive communication protocols based on machinery in the communication complexity literature, given a total privacy constraint. Barnes et al. \cite{barnes2019lower} studied a quantized Fisher information under blackboard communication protocols and proved a Bayesian lower bound using a continuous prior and the van Trees inequality, which typically requires the usage of the $\ell_2^2$ loss. This Fisher information based approach has been extended to distributed estimation under local differential privacy constraints in \cite{barnes2020privacy} recovering the results of \cite{duchi2019lower}  in the case of the $\ell_2^2$ loss. For the identity testing, there are three recent papers Amin et al. \cite{amin2020pan}, Berrett and Butucea \cite{berrett2020locally}, and Acharya et al. \cite{acharya2020interactive} which established the sharp statistical rates under interactive communication protocols, where they focused on the discrete distribution estimation model and sequential interactive protocols. Sequential communication protocols are stronger than $\Pi_{\mathsf{SMP}}$ but weaker than $\Pi_{\mathsf{BB}}$, as samples are encoded in a sequential fashion and the encoding of the sample $i$ can decode on the messages transmitted by sensors $1,\dots,i-1$. In this paper, we extend the results of  \cite{han2018geometric} to blackboard communication protocols $\Pi_{\mathsf{BB}}$ (fixing the mistake in \cite{han2018distributed} and \cite{han2018geometric}) via a generalization of the idea in \cite{acharya2020interactive}, and therefore extending the approach of \cite{acharya2020interactive} to a broader family of statistical models and the blackboard communication protocol. 

We also compare with a recent paper \cite{acharya2020general} which appears after our submission. Both papers build upon the framework presented for the discrete setting in \cite{acharya2020interactive}, and have similar assumptions and results on the high-dimensional estimation problem with communication constraints. We also point out some differences. In terms of scope, \cite{acharya2020general} also studied the privacy constraints, and its latest version analyzed sparse estimation models in more detail. In terms of assumptions, an exact orthogonality assumption of likelihood ratios is required in \cite{acharya2020general}, whereas our likelihood ratio condition could be viewed as an approximate version which enables us to study the logistic regression model as well. Finally, in terms of the communication protocol, our blackboard communication protocol is more general than the sequential communication protocol studied in \cite{acharya2020general}, requiring additional effort in handling the tree-based communication protocol.

\subsection{Notation}
For a finite set $A$, let $|A|$ denote its cardinality; $[n]\triangleq \{1,2,\cdots,n\}$; for a measure $\mu$, let $\mu^{\otimes n}$ denote its $n$-fold product measure; lattice operations $\wedge, \vee$ are defined as $a\wedge b=\min\{a,b\}, a\vee b=\max\{a,b\}$; throughout the paper, logarithms $\log(\cdot)$ are in the natural base; $\|P-Q\|_{\mathsf{TV}}$ and $D_{\mathsf{KL}}(P\|Q)$ denote the total variation (TV) distance and Kullback--Leibler (KL) divergence between probability measures $P$ and $Q$, respectively; $\mathsf{Multi}(n;P)$ denotes the multinomial model which observes $n$ independent samples from $P$; for a matrix $A\in \bR^{m\times n}$, $\|A\|_{\mathsf{op}} = \max_{x\in \bR^n: \|x\|_2=1} \|Ax\|_2$ denotes the operator norm; for non-negative sequences $\{a_n\}$ and $\{b_n\}$, the notation $a_n\lesssim b_n$ (or $b_n\gtrsim a_n, a_n=O(b_n), b_n=\Omega(a_n)$) means $\limsup_{n\to\infty} \frac{a_n}{b_n}<\infty$, and $a_n\ll b_n$ ($b_n\gg a_n, a_n=o(b_n), b_n=\omega(a_n)$) means $\limsup_{n\to\infty}\frac{a_n}{b_n}=0$, and $a_n\asymp b_n$ (or $a_n=\Theta(b_n)$) is equivalent to both $a_n\lesssim b_n$ and $b_n\lesssim a_n$.

\subsection{Organization}
The rest of the paper is organized as follows. Section \ref{sec:main} presents our assumptions on the statistical model and two main lower bounds on the distributed minimax risk, which lead to new results or recover existing results in distributed estimation. In Section \ref{sec.blackboard} we introduce the tree representation of the blackboard communication protocol, and sketch the lower bound proof based on this representation. Section \ref{sec.geo} is devoted to the proof of Theorems \ref{thm.general} and \ref{thm.sub-gaussian}, where the key steps are two geometric inequalities. Further discussions are in Section \ref{sec.discussion}, and auxiliary lemmas and the proof of main lemmas are in the appendices.

\section{Main Results}\label{sec:main}
\subsection{Assumptions}
To derive meaningful results in the general minimax formulation, proper assumptions are necessary for the statistical model $(P_\theta)_{\theta\in\Theta\subseteq \bR^d}$ and the loss function $L$. To this end, we begin with the standard regularity condition on $(P_\theta)_{\theta\in \Theta}$. 
\begin{assumption}\label{assump.ULAN}
	The statistical model $(P_\theta)_{\theta\in\Theta}$ is differentiable in quadratic mean at every $\theta\in \Theta$, with the score function $S_{\theta}$ and the Fisher information matrix $I_\theta$.  
\end{assumption}
Note that Assumption \ref{assump.ULAN} is a mild condition commonly used in classical asymptotic statistics \cite{ibragimov2013statistical}, which leads to the asymptotically tight Cram\'{e}r--Rao lower bound for centralized estimation. However, to obtain finite-sample results we need additional assumptions requiring the following notations. For a binary vector $u\in \{\pm 1\}^m$ and $j\in [m]$, let $u^{\oplus j}$ be the vector after flipping the $j$-th coordinate of $u$. Also, for two binary vectors $u, u'\in \{\pm 1\}^m$, let $d_{\mathsf{Ham}}(u,u') = \sum_{j=1}^m \mathbbm{1}(u_j\neq u_j')$ be their Hamming distance. In addition, let $\calX\subseteq \bR^d$ be the common support of all probability measures $(P_\theta)_{\theta\in\Theta}$, and $\calA$ be a generic action space in which the estimator $\widehat{\theta}$ takes value. Finally, by a loss function $L$ we mean a generic non-negative (measurable) function $L: \Theta\times \calA \to \bR_+$. The next assumption is a refinement of Assumption \ref{assump.ULAN} which concerns the finite-sample property of $(P_\theta)_{\theta\in\Theta}$ and $L$. 

\begin{assumption}\label{assump.regularity}
	There exist $1\le d_0\le d$ and a subset of parameters $(\theta_u)_{u\in \{\pm 1\}^{d_0}}\subseteq \Theta$ such that the following conditions hold: 
	\begin{enumerate}
		\item \textbf{Regular grid condition}: For each $u\in \{\pm 1\}^{d_0}$, the $d\times d_0$ matrix $M_u$ with columns $\theta_{u^{\oplus j}} - \theta_u$ ranging over $j\in [d_0]$ has an operator norm at most $2\delta$. 
		\item \textbf{Separation condition}: for any $u,u'\in \{\pm 1\}^{d_0}$, it holds that
		\begin{align}\label{eq:separation}
			\inf_{a\in\calA} \left[ L(\theta_u, a) + L(\theta_{u'},a) \right] \ge \kappa\cdot d_{\mathsf{Ham}}(u,u').
		\end{align}
		\item \textbf{Likelihood ratio condition}: for any $u\in \{\pm 1\}^{d_0}$ and $j\in [d_0]$, it holds that
		\begin{align}\label{eq:approximation}
			\bE_{X\sim P_{\theta_u}}\left[\left|\frac{dP_{\theta_{u^{\oplus j}}}}{dP_{\theta_u}}(X) - 1 - (\theta_{u^{\oplus j}} - \theta_u)^\top S_{\theta_u}(X)\right|^2\right] \le \varepsilon^2. 
		\end{align}
		In addition, it holds that ${dP_{\theta_{u^{\oplus j}}}}/{dP_{\theta_u}}(x)\ge 1/2$ for all $x\in \calX$.
	\end{enumerate}
	If all above conditions hold, we call this statistical estimation problem \textbf{$(d,d_0,\delta,\kappa,\varepsilon)$-regular}. 
\end{assumption}

Assumption \ref{assump.regularity} will be best understood via an important special case. Consider $d_0=d$, and for each $u\in \{\pm 1\}^{d}$, let $\theta_u = \theta_0 + \delta u$ be a local perturbation of some $\theta_0\in \Theta$. Then the regularity grid condition clearly holds as the matrix $M_u$ is diagonal with diagonal entries being $\pm 2\delta$. Therefore, this condition essentially says that the parameters $(\theta_u)_{u\in \{\pm 1\}^{d_0}}$ look like the vertices of a cube with side length $\delta$. The separation condition is standard in applying Assouad or Fano-type arguments to a cube-like hypothesis class \cite{yu1997assouad}, and is fulfilled for many natural loss functions with $\kappa=\kappa(\delta)$ a function of $\delta$. For example, $\kappa(\delta) = 2\delta^p$ when $L = \ell_p^p$, with important special cases including the $\ell_1$ loss when $p=1$, and the mean squared error when $p=2$. The last likelihood ratio condition is motivated by the local expansion
\begin{align*}
	\frac{dP_{\theta+t\cdot h}}{dP_{\theta}}(x) = \exp\left(t\cdot h^\top S_\theta(x) - \frac{t^2}{2}\cdot h^\top I_\theta h + o_{P_\theta}(t^2) \right), \quad \forall h\in \bR^d
\end{align*}
and $e^x = 1+x+x^2/2 + o(x^2)$, as well as the identity $\bE_{X\sim P_\theta}[S_\theta(X)S_\theta(X)^\top] = I_\theta$. In other words, \eqref{eq:approximation} is a quantitative way to approximate the local likelihood ratio by score functions, with the approximation error $\varepsilon = \varepsilon(\delta)$ typically growing with $\delta$. This quantitative condition will help us to show the indistinguishability among the locally perturbed statistical models. We also remark that the likelihood ratio is computed only between two neighboring vertices and often not growing with the dimensionality $d$, thus the lower bound assumption on the local likelihood ratio is not restrictive in high dimensions. Finally, to show that Assumption \ref{assump.regularity} holds for a certain statistical model, typically we first choose a suitable perturbation distance $\delta$ and then work out the parameters $\kappa$ and $\varepsilon$ as functions of $\delta$. 

While the first two conditions of Assumption \ref{assump.regularity} are relatively easier to hold, the last condition may fail for statistical models with an unbounded support (e.g. $\calX = \bR^d$). To mitigate this drawback, we propose a slightly weaker assumption which requires that the likelihood ratio condition holds with a high probability. 

\begin{assumption}\label{assump.regularity_weak}
	Assume the same conditions in Assumption \ref{assump.regularity}, except that in the likelihood ratio condition, there exists some $\calX_0\subseteq \calX$ such that
	\begin{align}\label{eq:approximation_weak}
		\bE_{X\sim P_{\theta_u}}\left[\left|\frac{dP_{\theta_{u^{\oplus j}}}}{dP_{\theta_u}}(X) - 1 - (\theta_{u^{\oplus j}} - \theta_u)^\top S_{\theta_u}(X)\right|^2\cdot \mathbbm{1}(X\in \calX_0)\right] \le \varepsilon^2,
	\end{align}
	and ${dP_{\theta_{u^{\oplus j}}}}/{dP_{\theta_u}}(x)\ge 1/2$ for all $x\in \calX_0$. Moreover, we require that $P_{\theta_u}(\calX_0)\ge 1-\alpha$ for all $u\in \{\pm 1\}^{d_0}$. If the above condition holds, we call this statistical estimation problem \textbf{$(d,d_0,\delta,\kappa,\varepsilon,\alpha)$-approximately-regular}.
\end{assumption}

In all the models considered in this paper, we always have $\alpha=o(n^{-1})$, so the likelihood ratio condition is only violated with a tiny probability, which suffices for our main Theorems \ref{thm.general} and \ref{thm.sub-gaussian}. The next proposition shows that many common statistical models are regular or approximately regular. 
\begin{proposition}\label{prop.assumption}
	For $L=\ell_p^p$ with $p\in [1,\infty)$, the following statements hold: 
	\begin{itemize}
		\item The product Bernoulli model $P_\theta = \prod_{j=1}^d \mathsf{Bern}(\theta_j)$ with $\Theta = [0,1]^d$ is $(d,d,\delta,\kappa,\varepsilon)$-regular with any $\delta\in (0,1/6)$, and $\kappa(\delta) = 2\delta^p, \varepsilon(\delta) \equiv 0$. 
		\item The product Bernoulli model $P_\theta = \prod_{j=1}^d \mathsf{Bern}(\theta_j)$ or the Multinomial model $P_\theta = \mathsf{Multi}(1;\theta)$ with $\Theta = \{(\theta_1,\cdots,\theta_d): \sum_{j=1}^d \theta_j = 1 \}$ is $(d,d/2,\delta,\kappa,\varepsilon)$-regular with any $\delta \in (0,1/(2d))$, and $\kappa(\delta)=2^{2-p}\delta^p$, $\varepsilon(\delta)\equiv 0$. 
		\item The Gaussian location model $P_\theta = \calN(\theta,\sigma^2I_d)$ with $\Theta = \bR^d$ is $(d,d,\delta,\kappa,\varepsilon,\alpha)$-approximately-regular with any $\delta\in (0,c\sigma/\sqrt{\log(nd)})$ for some small constant $c>0$, and $\kappa(\delta) = 2\delta^p$, $\varepsilon(\delta) = O(\delta^2/\sigma^2)$, $\alpha = o(n^{-1})$.
		\item Consider the following logistic regression model $P_\theta$ with random design: the observation vector is $X=(z,y)$, with feature $z\sim \calN(0,I_d)$ and label $y\sim \mathsf{Bern}(1/(1+\exp(-\theta^\top z)))$ given $z$. This model with $\Theta=\{\theta\in\bR^d: \|\theta\|_2\le 1 \}$ is $(d,d,\delta,\kappa,\varepsilon,\alpha)$-approximately-regular with any $\delta\in (0,1/\sqrt{d})$, and $\kappa(\delta)=2\delta^p$, $\varepsilon(\delta)=O(\delta^2)$, $\alpha=o(n^{-1})$. 
	\end{itemize}
\end{proposition}

\subsection{Main Theorems}
Although the previous assumptions give that the local likelihood ratio could be approximated by a linear form of the score function, they do not impose any assumption on the score function itself. As we recall from the classical asymptotic theory that the score function and the Fisher information play central roles in the estimation error, additional properties on the score function will be required to state the minimax lower bound. Our first and general lower bound states that, if the score function has a finite variance along any direction, then the estimation error in the distributed case decays at most exponentially with $k$. 
\begin{theorem}[General lower bound I]\label{thm.general}
	Let the statistical problem be $(d,d_0,\delta,\kappa,\varepsilon,\alpha)$-approximately-regular with 
	\begin{align*}
		I_0 \triangleq \max_{u\in \{\pm 1\}^{d_0}} \max_{v\in \bR^d: \|v\|_2 = 1} \bE_{\theta_u}[(v^\top S_{\theta_u}(X))^2] < \infty. 
	\end{align*}
	Then it holds that
	\begin{align*}
		\inf_{\Pi_{\mathsf{BB}}} \inf_{\widehat{\theta}}\sup_{\theta\in\Theta} \bE_{\theta}[L(\theta,\widehat{\theta})] \ge c\kappa d_0\left[\exp\left(-Cn\left(\frac{2^k\wedge d}{d_0}\cdot I_0\delta^2 + \varepsilon^2  \right)\right) - 4n\alpha\right], 
	\end{align*}
	where the infimum is taken over all possible estimators $\widehat{\theta}=\widehat{\theta}(Y^n)$ and blackboard protocols with $k$-bit communication constraint, and $c,C>0$ are absolute constants independent of $(n,d,k,I_0,d_0,\delta,\kappa,\varepsilon,\alpha)$. 
\end{theorem}

We first show how Theorem \ref{thm.general} could be used to give a meaningful lower bound. Since $\kappa=\kappa(\delta)$ and $\varepsilon=\varepsilon(\delta)$ are typically increasing in $\delta$, as $\delta$ increases, the leading coefficient will be larger while the exponential term will be smaller. To handle this tradeoff, we will choose the largest $\delta>0$ such that the statistical problem remains to be (approximately-)regular, and $\delta^2 = O(d_0/(nI_0(2^k\wedge d)))$. Meanwhile, for this choice of $\delta$, we expect that $\varepsilon=\varepsilon(\delta)$ is at most $O(n^{-1/2})$, which holds in many examples in the next subsection. Finally, for this choice of $\delta$, we conclude that the minimax risk is lower bounded by $\Omega(d_0\kappa(\delta))$. 

We provide several intuitive implications of Theorem \ref{thm.general}. Assume for simplicity that $L = \ell_2^2$, $d_0=d$, $\kappa=\kappa(\delta)=2\delta^2$ and $\varepsilon=\varepsilon(\delta)\equiv 0$, then the choice of $\delta^2 \asymp d/(nI_0(2^k\wedge d))$ in Theorem \ref{thm.general} leads to the lower bound $\Omega(d^2/(nI_0(2^k\wedge d)))$. In the centralized case without any communication constraints, we have $k=\infty$ and therefore the lower bound $\Omega(d/(nI_0))$ for the mean squared error. Since the Fisher information matrix $I_\theta$ satisfies $I_\theta = \bE_{\theta}[S_{\theta}(X)S_{\theta}(X)^\top]$, an equivalent expression of $I_0$ is
\begin{align*}
	I_0 =  \max_{u\in \{\pm 1\}^{d_0}} \lambda_{\max}(I_{\theta_u}), 
\end{align*}
where $\lambda_{\max}$ denotes the largest eigenvalue. As a comparison, the standard Cram\'{e}r--Rao lower bound for the mean squared error is $\Omega(\text{trace}(I_{\theta}^{-1})/n)$ for any $\theta\in\Theta$ \cite{Hajek1972local}. Consequently, Theorem \ref{thm.general} reduces to a weaker but non-asymptotic version of the Cram\'{e}r--Rao lower bound in the centralized case, which often remains rate-optimal when $P_\theta$ is of a product structure. 

Now what happens when there are communication constraints? Using the above result, in the most communication-starved case $k=1$, we have an effective sample size reduction from $n$ to $n/d$. This bound is intuitively achievable by a simple grouping idea: the sensors are splitted into $n/d$ groups, and all $d$ sensors in one group ``simulate" a full $d$-dimensional observation with each sensor working on one coordinate (see, e.g. Proposition \ref{prop.bernoulli}). Therefore, we expect that the dependence on $n,d$ of our lower bound to be tight for $k=1$. When $k>1$, the lower bound $\Omega(d^2/(nI_0(2^k\wedge d)))$ shows that the dependence of the squared $\ell_2$ risk on $k$ cannot be faster than $2^{-k}$, i.e., the penalty incurred by the distributed setting reduces at most exponentially in $k$. In the next subsection we will see examples where this exponential reduction is indeed tight. 

A natural question is that whether or not the exponential dependence on $k$ is always tight. The answer turns out to be \emph{negative}: the above penalty will reduce at most linearly in $k$ when the score function has a sub-Gaussian tail along any direction. Recall that the $\psi_2$-norm of a random variable $X$ is defined by 
$$\|X\|_{\psi_2(P)}=\inf\left\{t>0: \bE_P\left[\exp\left(\frac{X^2}{t^2}\right)\right]\le 2\right\},$$ 
which is the Orlicz norm of $X$ associated with the Orlicz function $\psi_2(x)=\exp(x^2)-1$ \cite{birnbaum1931verallgemeinerung}. There are some equivalent definitions of the $\psi_2$-norm, and $\|X\|_{\psi_2(P)}\le \sigma$ if and only if $X$ is sub-Gaussian under $P$ with parameter $\Theta(\sigma)$ \cite{vershynin2010introduction}. The following theorem shows another lower bound when the score function has a finite $\psi_2$-norm along any direction. 

\begin{theorem}[Lower bound with sub-Gaussian structure]\label{thm.sub-gaussian}
	Let the statistical problem be $(d,d_0,\delta,\kappa,\varepsilon,\alpha)$-approximately-regular with 
	\begin{align*}
		\Sigma_0 \triangleq \max_{u\in \{\pm 1\}^{d_0}} \max_{v\in \bR^d: \|v\|_2 = 1} \|v^\top S_{\theta_u}(X)\|_{\psi_2(P_{\theta_u})}^2 < \infty. 
	\end{align*}
	Then it holds that
	\begin{align*}
		\inf_{\Pi_{\mathsf{BB}}} \inf_{\widehat{\theta}}\sup_{\theta\in\Theta} \bE_{\theta}[L(\theta,\widehat{\theta})] \ge c\kappa d_0\left[\exp\left(-Cn\left(\frac{k\wedge d}{d_0}\cdot \Sigma_0\delta^2 + \varepsilon^2 \right)\right) - 4n\alpha\right], 
	\end{align*}
	where the infimum is taken over all possible estimators $\widehat{\theta}=\widehat{\theta}(Y^n)$ and blackboard protocols with $k$-bit communication constraint, and $c,C>0$ are absolute constants independent of $(n,d,k,\Sigma_0,d_0,\delta,\kappa,\varepsilon,\alpha)$. 
\end{theorem}

Using the similar intuitive analysis, Theorem \ref{thm.sub-gaussian} roughly shows a lower bound $\Omega(d^2/(n\Sigma_0(k\wedge d)))$ for the mean squared error. When the coordinates of the score function $S_\theta(X)$ are independent (which holds when $P_\theta$ is a product distribution), the quantity $\Sigma_0$ is essentially the maximum $\psi_2$ norm for each coordinate. Compared with the lower bound $\Omega(d^2/(nI_0(2^k\wedge d)))$ in Theorem \ref{thm.general}, the new lower bound has a better dependence on $k$ when the score function not only admits a finite variance but also behaves like a Gaussian random variable. However, neither of these bounds is better than the other in general, for it is possible that $I_0 \ll \Sigma_0$; the next subsection will provide examples where each of these bounds is tight. We also remark that the different dependence on $k$ in Theorems \ref{thm.general} and \ref{thm.sub-gaussian} is due to the nature of different geometric inequalities (cf. Lemma \ref{lemma.geometry_1} and Lemma \ref{lemma.geometry_2}) satisfied by general probability distributions and sub-Gaussian distributions. 

\subsection{Applications}
Next we apply Theorems \ref{thm.general} and \ref{thm.sub-gaussian} to some concrete statistical estimation examples. The first corollary concerns the discrete distribution estimation model. 

\begin{corollary}[Discrete distribution estimation]\label{cor.multi}
	Let $P_\theta=\mathsf{Multi}(1;\theta)$ with $\Theta$ being the probability simplex over $d$ elements. For $k\in \mathbb{N}$, $p\in [1,\infty)$, and $n\ge d^2/(2^k\wedge d)$, we have
	\begin{align*}
		\inf_{\Pi_{\mathsf{BB}}} \inf_{\widehat{\theta}}\sup_{\theta\in\Theta} \bE_{\theta}\|\widehat{\theta}-\theta\|_p^p \ge C_p\cdot \frac{d}{(n(2^k\wedge d))^{p/2}}, 
	\end{align*}
	where $C_p>0$ is an absolute constant independent of $n,k,d$.
\end{corollary}

Using the construction of $\theta_u$'s in Proposition \ref{prop.assumption}, after some algebra one could verify that $I_0 = O(d)$ in Theorem \ref{thm.general}. Consequently, choosing $\delta = c/\sqrt{n(2^k\wedge d)}$ for a small enough constant $c>0$ ensures that $\delta<1/(2d)$ (which is required in Proposition \ref{prop.assumption}), with $\kappa \asymp \delta^p$ and $\varepsilon = 0$, giving the result of Corollary \ref{cor.multi}. For $p\in \{1,2\}$, there is a matching upper bound in \cite{han2018distributed}, showing the tightness of this minimax lower bound. This result also improves over the total communication budget in \cite{diakonikolas2017communication}. Under the sequential communication protocol, the recent paper \cite{acharya2020general} established the same lower bound for $n\ge d^2/(2^k\wedge d)$, as well as a different lower bound $\Omega((n(2^k\wedge d))^{-(p-1)/2})$ for $n<d^2/(2^k\wedge d)$, the tightness of which is currently unclear (this lower bound also follows from Corollary \ref{cor.multi} via replacing $d$ by a smaller quantity $d_{\min}=\sqrt{n(2^k\wedge d)}$ such that $n\ge d_{\min}^2/(2^k\wedge d_{\min})$). Note that in this case, the tight dependence of the minimax risk on $k$ is exponential. 

The next corollary characterizes the distributed minimax risk of mean estimation in the Gaussian location model. 
\begin{corollary}[Gaussian location model]\label{cor.gaussian}
	Let $P_\theta=\calN(\theta,\sigma^2I_d)$ with $\Theta=\bR^d$. For $k\in \mathbb{N}$, $p\in [1,\infty)$ and $n\ge d^2/(k\wedge d)^2 + d\log d/(k\wedge d)$, we have
	\begin{align*}
		\inf_{\Pi_{\mathsf{BB}}} \inf_{\widehat{\theta}}\sup_{\theta\in\Theta} \bE_{\theta}\|\widehat{\theta}-\theta\|_p^p \ge C_p\cdot d\left(\frac{d\sigma^2}{n(d\wedge k)}\right)^{\frac{p}{2}},
	\end{align*}
	where $C_p>0$ is an absolute constant independent of $n,k,d,\sigma^2$.
\end{corollary}

For the Gaussian location model, the score function is $S_{\theta}(x) = (x-\theta)/\sigma^2$, and therefore the assumption of Theorem \ref{thm.sub-gaussian} is fulfilled with $\Sigma_0 = O(1/\sigma^2)$. Consequently, in Theorem \ref{thm.sub-gaussian} we may choose $\delta \asymp \sigma \sqrt{d/(n(d\wedge k))}$ (which satisfies the constraint of Proposition \ref{prop.assumption} by the choice of $n$) with $\kappa\asymp \delta^p$ and $\varepsilon \asymp \delta^2/\sigma^2 = O(n^{-1/2})$ to derive Corollary \ref{cor.gaussian}. Note that for $p=2$, Corollary \ref{cor.gaussian} recover the results in \cite{zhang2013information,garg2014communication}, without logarithmic factors in the risk. Also, in this model, the tight dependence of the minimax risk on $k$ is linear. 

The above two models have different tight dependence on $k$: in Corollary \ref{cor.multi}, when $2^k<d$, we see an effective sample size reduction from $n$ to $n2^k/d$; in Corollary \ref{cor.gaussian}, when $k<d$, we see an effective sample size reduction from $n$ to $nk/d$. This phenomenon may be better illustrated using the following example:

\begin{corollary}[Product Bernoulli model]\label{prop.bernoulli}
	Let $P_\theta=\prod_{i=1}^d \mathsf{Bern}(\theta_i)$. If $\Theta=[0,1]^d$ and $n\ge \frac{d}{d\wedge k}$, we have
	$$
	\inf_{\Pi_{\mathsf{BB}}}\inf_{\widehat{\theta}}\sup_{\theta\in\Theta} \bE_{\theta}\|\widehat{\theta}-\theta\|_2^2 \asymp \frac{d^2}{nk} \vee \frac{d}{n}. 
	$$
	If $\Theta\triangleq \{(\theta_1,\cdots,\theta_d)\subseteq [0,1]^d: \sum_{i=1}^d \theta_i=1 \}$ and $n\ge \frac{d^2}{d\wedge 2^k}$, we have
	$$
	\inf_{\Pi_{\mathsf{BB}}}\inf_{\widehat{\theta}}\sup_{\theta\in\Theta} \bE_{\theta}\|\widehat{\theta}-\theta\|_2^2 \asymp \frac{d}{n2^k} \vee \frac{1}{n}.
	$$
\end{corollary}

The first lower bound follows from Theorem \ref{thm.sub-gaussian}, and it was also obtained in \cite{zhang2013information} under the independent protocol with a matching upper bound. The same lower bound was also obtained in recent papers \cite{acharya2020interactive,acharya2020general}. The second lower bound follows from Theorem \ref{thm.general}, and the upper bound could be obtained using the ``simulate-and-infer" procedure in \cite{acharya2020inference2}. Note that the dependence of the squared $\ell_2$ risk on $k$ is significantly different under these two scenarios, even if both of them are product Bernoulli models: the dependence is linear in $k$ when $\Theta=[0,1]^d$, while it is exponential in $k$ when $\Theta$ is the probability simplex. We remark that this is due to the different behaviors of the score function: in the first case, we have $I_0\asymp \Sigma_0=\Theta(1)$; in the second case, we have $I_0\asymp d\ll d^2\asymp \Sigma_0$. Hence, Theorem \ref{thm.sub-gaussian} utilizes the sub-Gaussian nature and gives a better lower bound in the first case, and Theorem \ref{thm.general} becomes better in the second case where the tail of the score function is essentially not sub-Gaussian.

In addition to mean estimation, our main theorems also provide the following lower bound for parameter estimation in logistic regression. 
\begin{corollary}[Logistic regression]\label{cor.logistic}
	Consider the logistic regression model with random design formulated in Proposition \ref{prop.assumption}. For $k\in \mathbb{N}$, $p\in [1,\infty)$ and $n\ge d^2/(d\wedge k)$, we have
	\begin{align*}
		\inf_{\Pi_{\mathsf{BB}}} \inf_{\widehat{\theta}}\sup_{\theta\in\Theta} \bE_{\theta}\|\widehat{\theta}-\theta\|_p^p \ge C_p\cdot d\left(\frac{d}{n(d\wedge k)}\right)^{\frac{p}{2}},
	\end{align*}
	where $C_p>0$ is an absolute constant independent of $n,k,d$.
\end{corollary}

The proof of Corollary \ref{cor.logistic} follows from Proposition \ref{prop.assumption} and Theorem \ref{thm.sub-gaussian}. Specifically,  in logistic regression, the score function at $x=(z,y)$ is given by
\begin{align*}
	S_\theta(x) = S_\theta(y,z) = \left(y - \frac{1}{e^{-\theta^\top z}+1}\right)z, 
\end{align*}
which satisfies the sub-Gaussian condition of Theorem \ref{thm.sub-gaussian} with $\Sigma_0 = O(1)$ as the scalar parameter of $z$ always lies in $[-1,1]$. Consequently, choosing $\delta \asymp \sqrt{d/(n(d\wedge k))}$, as well as the quantities $\kappa(\delta)=2\delta^p$, $\varepsilon(\delta)=O(\delta^2)=O(n^{-1/2})$, and $\alpha=o(n^{-1})$ given by Proposition \ref{prop.assumption}, in Theorem \ref{thm.sub-gaussian} proves Corollary \ref{cor.logistic}. Note that the above argument only requires the random feature vector $z$ to be sub-Gaussian. For $p=2$, the same result was also proved in \cite{barnes2019minimax} using the van Trees inequality, with a matching upper bound when $z\sim \mathsf{Unif}(\{\pm 1\}^d)$. A similar lower bound for logistic regression under privacy constraint was obtained in \cite[Corollary 4]{duchi2019lower}: although they studied the excess risk instead of the $\ell_2^2$ loss, in the proof they essentially lower bounded the excess risk by the $\ell_2^2$ loss. For $p=2$, their overhead compared with the centralized case is $d/(\varepsilon\wedge \varepsilon^2)$ with average privacy budget $\varepsilon$, while ours is $d/(d\wedge k)$. Also, while the lower bound under the privacy constraint could be attained using private gradient updates (see \cite[Corollary 3.2]{bhowmick2018protection}), it is unknown whether a similar approach works under the communication constraint. 

Finally we look at the distributed mean estimation problem for sparse Gaussian location models. 
\begin{theorem}[Sparse Gaussian location model]\label{thm.sparse}
	Let $P_{\theta}=\calN(\theta,\sigma^2I_d)$ with $\Theta= \{\theta\in \bR^d:\|\theta_0\|\le s\}$ with $s\le d/2$. For $k\in \mathbb{N}$ and $n\ge sd^2\log(d/s)/(k\wedge d)^2$, we have
	\begin{align*}
		\inf_{\Pi_{\mathsf{SMP}}} \inf_{\widehat{\theta}}\sup_{\theta\in\Theta} \bE_{\theta}\|\widehat{\theta}-\theta\|_2^2 \ge C\cdot \left(\frac{sd\log(d/s)}{nk} \vee \frac{s\log(d/s)}{n}\right)\sigma^2
	\end{align*}
	where $C>0$ is an absolute constant independent of $n,d,s,k,\sigma^2$, and $\Pi_{\mathsf{SMP}}$ represents the family of simultaneous message passing protocols. 
\end{theorem}

Under a different notion of communication cost, \cite{braverman2016communication} proved a lower bound $\Omega(sd\sigma^2/(nk))$, \emph{without} the logarithmic factor $\log(d/s)$, under  blackboard communication protocols. Moreover, under the above notion of communication and $\Pi_{\mathsf{SMP}}$, an upper bound $O(sd\sigma^2\log d/(nk))$, \emph{with} the logarithmic factor, was obtained in \cite{garg2014communication}. Interestingly, the recent paper \cite{acharya2020general} showed that an upper bound of $O(sd\sigma^2/(nk))$, \emph{without} the logarithmic factor, is indeed achievable under the sequential communication protocol. Therefore, Theorem \ref{thm.sparse} shows that the logarithmic factor is unavoidable under the non-interactive communication protocols, so there is a strict separation between the interactive and non-interactive protocols. The existence/non-existence of the logarithmic factor in constrained sparse estimation is an interesting research topic, and has drawn several recent attentions such as \cite{acharya2020estimating, chen2021breaking}. 

Ignoring the issues on the logarithmic factor, we see that as opposed to the logarithmic dependence on the ambient dimension $d$ in the centralized setting, the number of nodes required to achieve a vanishing error in the distributed setting must scale with $d$. Hence, the sparse mean estimation problem becomes much harder in the distributed case, and the dimension involved in the effective sample size reduction (from $n$ to $nk/d$) is the ambient dimension $d$ instead of the effective dimension $s$.

\section{Representations of Blackboard Communication Protocol}\label{sec.blackboard}
The centralized lower bounds without communication constraints simply follows from the classical asymptotics \cite{Hajek1970characterization,Hajek1972local}, thus we devote our analysis to the communication constraints. In this section, we establish an equivalent tree representation of the blackboard communication protocol, and prove the statistical lower bound based on this representation.

\subsection{Tree representation of blackboard communication protocol}
Assume first that there is no public/private randomness, which will be revisited in the next subsection, and thus the protocol is deterministic. In this case, the blackboard communication protocol $\Pi_{\mathsf{BB}}$ can be viewed as a binary tree \cite{kushilevitz1997communication}, where each internal node $v$ of the tree  is assigned a deterministic label $l_v\in [n]$ indicating the identity of the sensor to write the next bit on the blackboard if the protocol reaches node $v$; 
the left and right edges departing from $v$ correspond to the two possible values of this bit and are labeled by $0$ and $1$ respectively.
Because all bits written on the blackboard up to the current time are observed by all nodes, the sensors can keep track of the progress of the protocol in the binary tree. The value of the bit written by node $l_v$ (when the protocol is at node $v$) can depend on the sample $X_{l_v}$ observed by this node (and implicitly on all bits previously written on the blackboard encoded in  the position of the node $v$ in the binary tree). Therefore, this bit can be represented by a binary function $a_v(x)\in \{0,1\}$, which we associate with the node $v$; sensor $l_v$ evaluates this function on its sample $X_{l_v}$ to determine the value of its bit. 

Note that the $k$-bit communication constraint for each node can be viewed as a labeling constraint for the binary tree; for each $i\in [n]$, each possible path from the root node to a leaf node can visit exactly $k$ internal nodes with label $i$. In particular, the depth of the binary tree is $nk$ and there is one-to-one correspondance between all possible transcripts $y\in \{0,1\}^{nk}$ and paths in the tree. Note that a proper labeling of the binary tree together with the collection of functions $\{a_v(\cdot)\}$ (where $v$ ranges over all internal nodes)  completely characterizes all possible (deterministic) communication strategies for the sensors. Under this protocol model, the distribution of the transcript $Y$ is 
\begin{align*}
	\mathbb{P}_{X_1,\cdots,X_n\sim P}(Y=y) = \mathbb{E}_{X_1,\cdots,X_n\sim P}\prod_{v\in \tau(y)} b_{v,y}(X_{l_v})
\end{align*}
where $v\in \tau(y)$ ranges over all internal nodes in the path $\tau(y)$ corresponding to $y\in \{0,1\}^{nk}$, and $b_{v,y}(x)=a_v(x)$ if the path $\tau(y)$ goes through the right child of $v$ and $b_{v,y}(x)=1-a_v(x)$ otherwise. Due to the independence of $X_1,\cdots,X_n$, we have the following lemma which is similar to the ``cut-paste" property \cite{bar2004information} for the blackboard communication protocol: 
\begin{lemma}\label{lemma.cut-paste}
	The distribution of the transcript $Y$ can be written as follows: for any $y\in \{0,1\}^{nk}$, we have
	$$
	\mathbb{P}_{X_1,\cdots,X_n\sim P}(Y=y) = \prod_{i=1}^n \mathbb{E}_{P}[p_{i,y}(X_i)]
	$$
	where $p_{i,y}(x)\triangleq \prod_{v\in \tau(y), l_v=i} b_{v,y}(x)$. 
\end{lemma}

The $k$-bit communication constraint results in the following important property: 
\begin{lemma}\label{lemma.total_weight}
	For each $i\in [n]$ and $\{x_j\}_{j=1}^n\in \calX^n$, the following equalities hold: 
	\begin{align*}
		\sum_{y\in \{0,1\}^{nk}} \prod_{j=1}^n p_{j,y}(x_j) = 1, \qquad \sum_{y\in \{0,1\}^{nk}} \prod_{j\neq i} p_{j,y}(x_j) = 2^k. 
	\end{align*}
\end{lemma}

\subsection{Minimax lower bound}
This subsection is devoted to setting up the proof of the minimax lower bound in Theorems \ref{thm.general} and \ref{thm.sub-gaussian}. To this end, we apply the standard testing argument with the Assouad's lemma \cite{assouad1983deux} to the cube-like distribution family $(P_{\theta_u})_{u\in \{\pm 1\}^{d_0}}$ in Assumptions \ref{assump.regularity} and \ref{assump.regularity_weak}, and arrive at a target quantity to be upper bounded in Section \ref{sec.geo}. This subsection is devoted exclusively to $(d,d_0,\delta,\kappa,\varepsilon)$-regular problems to reflect the main ideas, while the modification to handle approximately regular problems is postponed to the next subsection. 

Let $U\sim \mathsf{Unif}(\{ \pm 1\}^{d_0})$, and write $P_u$ as a shorthand of $P_{\theta_u}$ throughout this section. Given the i.i.d. observations $X_1,\cdots,X_n\sim P_u$ and a communication protocol $\Pi$, let $Q_u$ be the probability distribution of the final transcript $Y\in \{0,1\}^{nk}$. As the final estimator $\widehat{\theta} = \widehat{\theta}(Y)$ is a function of $Y$, the standard separation condition \eqref{eq:separation} with the Assouad's lemma \cite{assouad1983deux} (see also \cite[Theorem 5]{han2019online}) gives that
\begin{align}\label{eq:assouad}
	\bE_U \bE_{Q_U}[L(\theta_U, \widehat{\theta}(Y))] \ge \frac{d_0\kappa}{2}\left(1-\frac{1}{d_0}\sum_{j=1}^{d_0}\bE_U\|Q_U - Q_{U^{\oplus j}}\|_{\mathsf{TV}} \right).
\end{align}
As \cite[Lemma 2.6]{Tsybakov2008} shows that $\|P-Q\|_{\mathsf{TV}}\le 1-\exp(-D_{\mathsf{KL}}(P\|Q))/2$, the above inequality \eqref{eq:assouad} together with the convexity of $x\mapsto \exp(-x)$ implies that
\begin{align}\label{eq:assouad_KL}
	\bE_U \bE_{Q_U}[L(\theta_U, \widehat{\theta}(Y))] &\ge \frac{d_0\kappa}{2}\cdot \frac{1}{2d_0}\sum_{j=1}^{d_0}\bE_U \exp\left(-D_{\mathsf{KL}}(Q_U \| Q_{U^{\oplus j}})\right) \nonumber \\
	&\ge \frac{d_0\kappa}{4}\cdot \exp\left(-\bE_U\left[\frac{1}{d_0}\sum_{j=1}^{d_0}D_{\mathsf{KL}}(Q_U \| Q_{U^{\oplus j}})\right]\right). 
\end{align}
The usage of (a slightly different form of) the inequality \eqref{eq:assouad_KL} is motivated by \cite{acharya2020interactive}, which studies discrete distribution estimation models under the sequential communication protocol. In the sequel, we extend this approach to generic statistical models and fully interactive (blackboard) communication protocols. 

Next we upper bound the average KL divergence in \eqref{eq:assouad_KL} for each given $U\in \{\pm 1\}^{d_0}$. To this end, first we note that it suffices to assume no private/public randomness due to the data-processing property of the KL divergence $D_{\mathsf{KL}}(P\|Q)\le \bE_R[D_{\mathsf{KL}}(P_{|R}\|Q_{|R})]$. Then by Lemma \ref{lemma.cut-paste}, we have
\begin{align*}
	Q_U(y) = \prod_{i=1}^n \bE_{X_i\sim P_U}[p_{i,y}(X_i)]
\end{align*}
for each transcript $y\in \{0,1\}^{nk}$. Consequently, for each $j\in [d_0]$, 
\begin{align*}
	D_{\mathsf{KL}}(Q_U \| Q_{U^{\oplus j}}) &= \sum_{i=1}^n  \sum_{y\in \{0,1\}^{nk}} \left(\prod_{s=1}^n \bE_{X_s\sim P_U}[p_{s,y}(X_s)] \right)\cdot \log \frac{\bE_{X_i\sim P_U}[p_{i,y}(X_i)]}{\bE_{X_i\sim P_{U^{\oplus j}}}[p_{i,y}(X_i)]} \\
	&\stepa{\le} \sum_{i=1}^n  \sum_{y\in \{0,1\}^{nk}} \left(\prod_{s=1}^n \bE_{X_s\sim P_U}[p_{s,y}(X_s)] \right)\cdot \left(\frac{\bE_{X_i\sim P_U}[p_{i,y}(X_i)]}{\bE_{X_i\sim P_{U^{\oplus j}}}[p_{i,y}(X_i)]} -1 \right) \\
	&\stepb{=} \sum_{i=1}^n  \sum_{y\in \{0,1\}^{nk}} \left(\prod_{s\neq i} \bE_{X_s\sim P_U}[p_{s,y}(X_s)] \right)\cdot \left(\frac{(\bE_{X_i\sim P_U}[p_{i,y}(X_i)] - \bE_{X_i\sim P_{U^{\oplus j}}}[p_{i,y}(X_i)])^2}{\bE_{X_i\sim P_{U^{\oplus j}}}[p_{i,y}(X_i)]} \right) \\
	&\stepc{\le} 2\sum_{i=1}^n  \sum_{y\in \{0,1\}^{nk}} \left(\prod_{s\neq i} \bE_{X_s\sim P_U}[p_{s,y}(X_s)] \right)\cdot \left(\frac{(\bE_{X_i\sim P_U}[p_{i,y}(X_i)] - \bE_{X_i\sim P_{U^{\oplus j}}}[p_{i,y}(X_i)])^2}{\bE_{X_i\sim P_{U}}[p_{i,y}(X_i)]} \right) \\
	&\stepd{=} 2\sum_{i=1}^n  \sum_{y\in \{0,1\}^{nk}} \left(\prod_{s\neq i} \bE_{X_s\sim P_U}[p_{s,y}(X_s)] \right)\cdot \frac{(\bE_{X_i\sim P_{U}}[p_{i,y}(X_i)(1 - dP_{U^{\oplus j}}/dP_{U}(X_i)) ])^2}{\bE_{X_i\sim P_{U}}[p_{i,y}(X_i)]}
\end{align*}
where (a) is due to the inequality $\log x\le x-1$, (b) follows from the identity 
\begin{align*}
	\sum_{y\in \{0,1\}^{nk}} \prod_{s=1}^n \bE_{X_s\sim P_U}[p_{s,y}(X_s)] = \sum_{y\in \{0,1\}^{nk}} \bE_{X_i\sim P_{U^{\oplus j}}}[p_{i,y}(X_i)] \cdot \prod_{s\neq i} \bE_{X_s\sim P_U}[p_{s,y}(X_s)] = 1
\end{align*}
given by Lemma \ref{lemma.total_weight}, (c) follows from the likelihood ratio condition in Assumption \ref{assump.regularity}, and (d) is due to a simple change of measure. Consequently, for each realization of $U$ we have
\begin{align}\label{eq:final_bound}
	\frac{1}{d_0}\sum_{j=1}^{d_0} D_{\text{KL}}(Q_U \| Q_{U^{\oplus j}}) &\le \frac{2}{d_0}\sum_{i=1}^n \sum_{y\in \{0,1\}^{nk}} \left(\prod_{s\neq i} \bE_{X_s\sim P_U}[p_{s,y}(X_s)] \right)\cdot  \frac{\|\bE_{X_i\sim P_U}[p_{i,y}(X_i)s_U(X_i)] \|_2^2}{\bE_{X_i\sim P_{U}}[p_{i,y}(X_i)]}, 
\end{align}
where $s_U(x)$ is a $d_0$-dimensional vector of likelihood ratios: 
\begin{align*}
	s_U(x) \triangleq \left(1 - \frac{dP_{U^{\oplus 1}}}{dP_U}(x), \cdots, 1 - \frac{dP_{U^{\oplus d_0}}}{dP_U}(x)\right). 
\end{align*}

To deal with $s_U(x)$, we use the likelihood ratio condition in Assumption \ref{assump.regularity} to write that 
\begin{align*}
	s_U(x) = -M_u^\top S_{\theta_U}(x) + \varepsilon(x),
\end{align*}
where $M_u$ is the matrix appearing in the regular grid condition in Assumption \ref{assump.regularity}, and $\varepsilon(x)$ is some remainder term satisfying that $\bE[\|\varepsilon(X)\|_2^2] \le d_0\varepsilon^2$ for all $X\sim P_U$. Consequently, for the remainder term we have
\begin{align*}
	& \sum_{y\in \{0,1\}^{nk}} \left(\prod_{s\neq i} \bE_{X_s\sim P_U}[p_{s,y}(X_s)] \right)\cdot  \frac{\|\bE_{X_i\sim P_U}[p_{i,y}(X_i)\varepsilon(X_i)] \|_2^2}{\bE_{X_i\sim P_{U}}[p_{i,y}(X_i)]} \\
	&\stepa{\le} \sum_{y\in \{0,1\}^{nk}} \left(\prod_{s\neq i} \bE_{X_s\sim P_U}[p_{s,y}(X_s)] \right)\cdot  \bE_{X_i\sim P_U}[p_{i,y}(X_i)\|\varepsilon(X_i)\|_2^2]\\
	&\stepb{=} \bE_{X_1,\cdots,X_n\sim P_U}\left[\sum_{y\in \{0,1\}^{nk}}\left(\prod_{s=1}^n p_{s,y}(X_s) \right)\cdot \|\varepsilon(X_i)\|_2^2 \right] \\
	&\stepc{=} \bE_{X_i\sim P_U}[\|\varepsilon(X_i)\|_2^2] \le d_0\varepsilon^2,
\end{align*}
where (a) is due to Cauchy--Schwarz, (b) swaps the expectation and sum, and (c) is due to the first identity of Lemma \ref{lemma.total_weight}. As for the main term, since $\|Ax\|_2\le \|A\|_{\mathsf{op}}\|x\|_2$, the regular grid assumption in Assumption \ref{assump.regularity} gives
\begin{align*}
	\frac{\|\bE_{X_i\sim P_U}[p_{i,y}(X_i)\cdot M_u^\top S_{\theta_U}(X_i)] \|_2^2}{\bE_{X_i\sim P_{U}}[p_{i,y}(X_i)]} \le 4\delta^2\cdot \frac{\|\bE_{X_i\sim P_U}[p_{i,y}(X_i) S_{\theta_U}(X_i)] \|_2^2}{\bE_{X_i\sim P_{U}}[p_{i,y}(X_i)]}. 
\end{align*}

Consequently, by the triangle inequality $\|x+y\|_2^2\le 2(\|x\|_2^2 + \|y\|_2^2)$ and Lemma \ref{lemma.total_weight}, the above inequalities together with \eqref{eq:assouad_KL} and \eqref{eq:final_bound} imply that
\begin{align}\label{eq:minimax_final}
	\bE_U \bE_{Q_U}[L(\theta_U, \widehat{\theta}(Y))]  \ge \frac{d_0\kappa}{4}\cdot \exp\left(-\frac{16nS}{d_0}\cdot \delta^2 - 4n\varepsilon^2\right), 
\end{align}
where 
\begin{align}\label{eq:target}
	S \triangleq \max_{u\in \{\pm 1\}^{d_0}}\max_{i\in [n]} \sum_{y\in \{0,1\}^{nk}} \left(\prod_{s\neq i} \bE_{X_s\sim P_u}[p_{s,y}(X_s)] \right)\cdot  \frac{\|\bE_{X_i\sim P_u}[p_{i,y}(X_i)S_{\theta_u}(X_i)] \|_2^2}{\bE_{X_i\sim P_u}[p_{i,y}(X_i)]}. 
\end{align}
Hence, to obtain the final minimax lower bound, it suffices to provide upper bounds of the quantity $S$ in \eqref{eq:target}. This is the main focus of Section \ref{sec.geo}. 

\subsection{Approximately regular problems}
In this subsection we show how to modify the above arguments to work for approximately regular problems. Note that when $\alpha=\omega(n^{-1})$, the lower bounds in Theorems \ref{thm.general} and \ref{thm.sub-gaussian} are negative and thus trivial; in the sequel we always assume that $\alpha=O(n^{-1})=o(1)$. For each $u\in \{\pm 1\}^{d_0}$, let $\widetilde{P}_u(\cdot) = P_u(\cdot \mid \calX_0)$ be the restriction of $P_u$ to the set $\calX_0$, and $\widetilde{Q}_u$ be the distribution of the transcript $Y$ under $X_1,\cdots,X_n\sim \widetilde{P}_u$. Then by the property of $\calX_0$ in Assumption \ref{assump.regularity_weak}, we have
\begin{align*}
	\max_{u\in \{\pm 1\}^{d_0}} \| \widetilde{Q}_u - Q_u \|_{\mathsf{TV}} &\le \max_{u\in \{\pm 1\}^{d_0}} \| \widetilde{P}_u^{\otimes n} - P_u^{\otimes n} \|_{\mathsf{TV}} \\
	&\le n\cdot  \max_{u\in \{\pm 1\}^{d_0}} \| \widetilde{P}_u- P_u\|_{\mathsf{TV}}\\
	&= n\cdot \max_{u\in \{\pm 1\}^{d_0}} P_u(\calX_0^c) \\
	&\le n\alpha. 
\end{align*}
Consequently, applying the triangle inequality to the TV distance in \eqref{eq:assouad} gives that
\begin{align*}
	\bE_U \bE_{Q_U}[L(\theta_U, \widehat{\theta}(Y))] \ge \frac{d_0\kappa}{2}\left(1-\frac{1}{d_0}\sum_{j=1}^{d_0}\bE_U\|\widetilde{Q}_U - \widetilde{Q}_{U^{\oplus j}}\|_{\mathsf{TV}} - 2n\alpha \right).
\end{align*}
Hence, if we could show that the new statistical model $(\widetilde{P}_u)_{u\in \{\pm 1\}^{d_0}}$ is regular with essentially the same parameters in Assumption \ref{assump.regularity}, and that the quantities $I_0$ and $\Sigma_0$ in Theorems \ref{thm.general} and \ref{thm.sub-gaussian} does not change much as we move from $P_u$ to $\widetilde{P}_u$, we could repeat the same analysis in the previous subsection with $(P_U,Q_U)$ replaced by $(\widetilde{P}_U,\widetilde{Q}_U)$ and arrive at the same results. 

To verify Assumption \ref{assump.regularity}, note that the regular grid assumption and separation condition do not depend on the statistical model and thus hold under $\widetilde{P}_u$ as well. For the likelihood ratio condition, note that for all $x\in \calX_0$, we have
\begin{align*}
	\frac{d\widetilde{P}_{u^{\oplus j}}}{d\widetilde{P}_u}(x) = \frac{dP_{u^{\oplus j}}}{dP_u}(x) \cdot \frac{P_{u}(\calX_0)}{P_{u^{\oplus j}}(\calX_0)}. 
\end{align*}
Therefore, the lower bound on the likelihood ratio could be replaced by $(1-\alpha)/2$, only slightly smaller than $1/2$. Moreover, as $dP_{u^{\oplus j}}/dP_u(x)\le 2$ for all $x\in \calX_0$, by triangle inequality
\begin{align*}
	&\bE_{X\sim \widetilde{P}_u}\left[\left|\frac{d\widetilde{P}_{\theta_{u^{\oplus j}}}}{d\widetilde{P}_{\theta_u}}(X) - 1 - (\theta_{u^{\oplus j}} - \theta_u)^\top S_{\theta_u}(X)\right|^2\right] \\
	&\le \frac{2}{1-\alpha}\bE_{X\sim {P}_u}\left[\left|\frac{d{P}_{\theta_{u^{\oplus j}}}}{d{P}_{\theta_u}}(X) - 1 - (\theta_{u^{\oplus j}} - \theta_u)^\top S_{\theta_u}(X)\right|^2\cdot \mathbbm{1}(X\in \calX_0)\right] + 2\cdot \left(\frac{2\alpha}{1-\alpha}\right)^2, 
\end{align*}
therefore the condition \eqref{eq:approximation_weak} implies that the parameter $\varepsilon$ in \eqref{eq:approximation} could simply be replaced by $O(\varepsilon/\sqrt{1-\alpha}+\alpha)$. As $\alpha=o(1)$, the new statistical model becomes regular with essentially the same parameters. 

To compute the new $I_0$ and $\Sigma_0$ under new models, note that for any non-negative function $f$, it holds that 
\begin{align*}
	\bE_{\widetilde{P}_u}[f(X)] \le \frac{1}{P_u(\calX_0)}\cdot \bE_{{P}_u}[f(X)] \le \frac{1}{1-\alpha}\cdot \bE_{{P}_u}[f(X)]. 
\end{align*}
Consequently, the quantities $I_0$ and $\Sigma_0$ in Theorems \ref{thm.general} and \ref{thm.sub-gaussian} for regular problems could be replaced by slightly larger quantities $I_0/(1-\alpha)$ and $\Sigma_0/(1-\alpha)$, respectively.

Combining the above points, the minimax lower bounds for approximately regular problems could be argued in an entirely similar manner as regular problems. 

\section{Lower Bounds via Geometric Inequalities}\label{sec.geo}
In this section, we upper bound the quantity $S$ in \eqref{eq:target} using two different geometric inequalities, and complete the proof of main Theorems \ref{thm.general} and \ref{thm.sub-gaussian}. 

\subsection{Proof of Theorem \ref{thm.general} via Geometric Inequality I}
Note that under a deterministic protocol, each function $p_{i,y}$ only takes value in $\{0,1\}$. Therefore, if we write $\calX_{i,y} = \{x\in\calX: p_{i,y}(x) = 1 \}$, then 
\begin{align*}
	\frac{\|\bE_{X_i\sim P_u}[p_{i,y}(X_i)S_{\theta_u}(X_i)] \|_2^2}{\bE_{X_i\sim P_u}[p_{i,y}(X_i)]} = P_u(\calX_{i,y})\cdot \|\bE_{P_u}[S_{\theta_u}(X) \mid \calX_{i,y}] \|_2^2. 
\end{align*}
Therefore, a quantity of interest is the $\ell_2$-norm of the conditional mean of a random vector $S_{\theta_u}(X)$ restricted to some set $\calX_{i,y}$. This motivates us to ask the following general question: 
\begin{question}\label{question_1}
	For a random vector $X\sim P$ and a target probability $P(A)=t\in (0,1)$, which subset $A\subseteq\calX$ maximizes the $\ell_2$ norm of the vector $\bE[X\mid A]$? What is the corresponding maximum $\ell_2$ norm?
\end{question}

The following lemma presents an answer to Question \ref{question_1} under the assumption that $X$ has a finite second moment along any direction. 
\begin{lemma}[Geometric Inequality I]\label{lemma.geometry_1}
	Assume that $\bE[(u^\top X)^2]\le I_0$ for all unit vector $u\in \bR^d$. Then for any set $A\subseteq \calX$, the following inequality holds:
	\begin{align*}
		\|\bE[X\mid A]\|_2^2 \le I_0\cdot \frac{1}{P(A)}.
	\end{align*}
	Moreover, the RHS could be improved to $I_0\cdot (1-P(A))/P(A)$ if $\bE[X]=0$. 
\end{lemma}

Note that Lemma \ref{lemma.geometry_1} is a dimension-free result: the LHS depends on the dimensionality $d$, while the RHS does not. For a comparison, if we trivially use $\|\bE[X\mid A]\|_2^2\le \bE[\|X\|_2^2\mid A]$, there would be an additional factor of $d$ on the RHS. The key observation in the dimensionality reduction is that the ``independence" between coordinates of $X$ is implied by the condition and needs to be exploited.

Now we have all necessary tools for the proof of Theorem \ref{thm.general}. Applying Lemma \ref{lemma.geometry_1} to the score function $S_{\theta_u}(x)$ in Theorem \ref{thm.general}, we have
\begin{align*}
	\frac{\|\bE_{X_i\sim P_u}[p_{i,y}(X_i)S_{\theta_u}(X_i)] \|_2^2}{\bE_{X_i\sim P_u}[p_{i,y}(X_i)]} \le I_0. 
\end{align*}
Consequently, by Lemma \ref{lemma.total_weight}, it holds that
\begin{align*}
	S \le I_0\cdot \max_{u\in \{\pm 1\}^{d_0}}\max_{i\in [n]}\sum_{y\in \{0,1\}^{nk}} \left(\prod_{s\neq i} \bE_{X_s\sim P_u}[p_{s,y}(X_s)] \right) = I_0\cdot 2^k,
\end{align*}
and plugging this upper bound of $S$ into the minimax lower bound \eqref{eq:minimax_final} completes the proof of one lower bound of Theorem \ref{thm.general}. For the other lower bound independent of $k$, an alternative upper bound of $S$ could be used: 
\begin{align*}
	S&\stepa{\le} \max_{u\in \{\pm 1\}^{d_0}}\max_{i\in [n]} \sum_{y\in \{0,1\}^{nk}} \left(\prod_{s\neq i} \bE_{X_s\sim P_u}[p_{s,y}(X_s)] \right)\cdot \bE_{X_i\sim P_u}[p_{i,y}(X_i)\cdot\|S_{\theta_u}(X_i)\|_2^2] \\
	&\stepb{=}\max_{u\in \{\pm 1\}^{d_0}}\max_{i\in [n]} \bE\left[\sum_{y\in \{0,1\}^{nk}} \left(\prod_{s=1}^n p_{s,y}(X_s) \right)\cdot\|S_{\theta_u}(X_i)\|_2^2\right] \\
	&\stepc{=} \max_{u\in \{\pm 1\}^{d_0}}\max_{i\in [n]} \bE_{X_i\sim P_u}[\|S_{\theta_u}(X_i)\|_2^2] \stepd{\le} dI_0,
\end{align*}
where (a) is due to Cauchy--Schwarz, (b) follows from swapping the expectation and the sum with the expectation taken over i.i.d. $X_1,\cdots,X_n\sim P_u$, (c) is due to Lemma \ref{lemma.total_weight}, and (d) follows from choosing $u=e_1,\cdots,e_d$, the canonical vectors, in the assumption of Theorem \ref{thm.general}. 

\subsection{Proof of Theorem \ref{thm.sub-gaussian} via Geometric Inequality II}
In this section, we provide another upper bound on $\|\bE[X\mid A]\|_2^2$ when $X$ is a sub-Gaussian random variable along any direction. 
\begin{lemma}[Geometric Inequality II]\label{lemma.geometry_2}
	Assume that $\| u^\top X\|_{\psi_2}^2 \le \Sigma_0$ for all unit vector $u\in \bR^d$. Then for any set $A\subseteq \calX$, the following inequality holds:
	\begin{align*}
		\|\bE[X\mid A]\|_2^2 \le \Sigma_0\cdot \log\frac{2}{P(A)}.
	\end{align*}
\end{lemma}

Note that lemma \ref{lemma.geometry_2} presents a dimension-free upper bound again. Compared with Lemma \ref{lemma.geometry_1}, Lemma \ref{lemma.geometry_2} improves the upper bound from $O(\Sigma_0)$ to $O(\Sigma_0 t\log\frac{1}{t})$ for sub-Gaussian random vector $X$, where $t=P(A)$ is the volume of the set $A$ and $\Sigma_0$ is the sub-Gaussian parameter. Lemma \ref{lemma.geometry_2} could be derived from standard arguments of the Talagrand's transportation-cost inequality \cite[Chapter 6]{ledoux2005concentration}, but for completeness we provide two proofs of Lemma \ref{lemma.geometry_2} in the appendix. The first proof directly reduces the problem to one dimension and then makes use of the Orlicz norm condition. The second proof is more geometric when $X$ is exactly Gaussian, where tight constants are obtained for $X\sim \mathsf{Unif}(\{\pm 1\}^d)$ via information-theoretic inequalities, and then the ``tensor power trick" is applied to prove the Gaussian case. 

To move from Lemma \ref{lemma.geometry_2} to an upper bound of the quantity $S$ and therefore Theorem \ref{thm.sub-gaussian}, note that the assumption of Theorem \ref{thm.sub-gaussian} and Lemma \ref{lemma.geometry_2} show that
\begin{align*}
	\frac{\|\bE_{X_i\sim P_u}[p_{i,y}(X_i)S_{\theta_u}(X_i)] \|_2^2}{\bE_{X_i\sim P_u}[p_{i,y}(X_i)]} \le \Sigma_0\cdot \bE_{X_i\sim P_u}[p_{i,y}(X_i)] \log\frac{2}{\bE_{X_i\sim P_u}[p_{i,y}(X_i)]}. 
\end{align*}
Therefore, 
\begin{align*}
	S&\le \Sigma_0\cdot \max_{u\in \{\pm 1\}^{d_0}}\max_{i\in [n]}\sum_{y\in \{0,1\}^{nk}} \left(\prod_{s=1}^n \bE_{X_s\sim P_u}[p_{s,y}(X_s)] \right)\log\frac{2}{\bE_{X_i\sim P_u}[p_{i,y}(X_i)]} \\
	&\stepa{\le} \Sigma_0\cdot \max_{u\in \{\pm 1\}^{d_0}}\max_{i\in [n]}\log\left( \sum_{y\in \{0,1\}^{nk}} \left(\prod_{s=1}^n \bE_{X_s\sim P_u}[p_{s,y}(X_s)] \right)\cdot \frac{2}{\bE_{X_i\sim P_u}[p_{i,y}(X_i)]}\right) \\
	&\stepb{=} \Sigma_0\cdot \max_{u\in \{\pm 1\}^{d_0}}\max_{i\in [n]} \log(2^{k+1}) = (k+1)\Sigma_0\log 2,
\end{align*}
where (a) is due to the first identity of Lemma \ref{lemma.total_weight} as well as the concavity of $x\mapsto \log x$, and (b) is due to the second identity of Lemma \ref{lemma.total_weight}. Plugging this upper bound into the minimax lower bound \eqref{eq:minimax_final} completes the proof of Theorem \ref{thm.sub-gaussian} (the other independent-of-$k$ upper bound of $S$ could be obtained analogously to the last section). 

\section{Discussions}\label{sec.discussion}
\subsection{Some Applications of Geometric Inequalities}
The inequalities in Lemmas \ref{lemma.geometry_1} and \ref{lemma.geometry_2} have some other combinatorial applications related to geometry. We consider the following combinatorial problem on the binary Hamming cube $\Omega=\{\pm1\}^d$:
\begin{enumerate}
	\item Suppose we pick half of the vectors in $\Omega$ and compute the mean $\bar{v}\in \bR^d$, i.e., $\bar{v}=|A|^{-1}\sum_{v\in A} v$ for some $A\subseteq \Omega, |A|=2^{d-1}$, what is the maximum possible $\ell_2$ norm $\|\bar{v}\|_2$?
	\item Suppose we pick $2^{dR}$ vectors in $\Omega$ and compute the mean $\bar{v}\in \bR^d$, where $R\in (0,1)$, what is the dependence of the maximum possible $\ell_2$ norm $\|\bar{v}\|_2$ on $d$ and $R$?
\end{enumerate}
This geometric problem is closely related to the optimal data compression in multiterminal statistical inference \cite{amari2011optimal}. We prove the following proposition:
\begin{proposition}\label{prop.geometry}
	Under the previous setting, we have
	\begin{align*}
		\max_{A\subseteq \Omega: |A|=2^{d-1}} \left\| \frac{1}{|A|}\sum_{v\in A}v \right\|_2 &= 1, \\
		\max_{A\subseteq \Omega: |A|=2^{dR}} \left\| \frac{1}{|A|}\sum_{v\in A}v \right\|_2 &= \sqrt{d}(1-2h_2^{-1}(R)) \cdot (1+o_d(1)), 
	\end{align*}
	where $h_2(\cdot)$ is the binary entropy function defined in Lemma \ref{lemma.h_2}. 
\end{proposition}

Proposition \ref{prop.geometry} gives the exact maximum $\ell_2$ norm when $|A|=2^{d-1}$ and its asymptotic behavior on $d$ and $R$ as $d\to\infty$ when $|A|=2^{dR}$. We see that for $|A|=2^{d-1}$, the maximum $\ell_2$ norm is attained when $A$ is the half space (or the $d-1$ dimensional sub-cube), i.e., $A=\{x\in \Omega: x_1=1\}$. However, for relatively small $|A|=2^{dR}$, the maximum $\ell_2$ norm is nearly attained at spherical caps, i.e., $A=\{x\in \Omega: d_{\mathsf{Ham}}(x,x_0)\le t \}$ for any fixed $x_0\in \Omega$ and a proper radius $t$ such that $|A|=2^{dR}$. Hence, there are different behaviors for dense and sparse sets $A$.

\subsection{Comparison with Strong Data Processing Inequalities (SDPI)}

We compare our techniques with existing ones in establishing the lower bound for distributed parameter estimation problem. By Fano's inequality, the key step is to upper bound the mutual information $I(U;Y)$ under the Markov chain $U-X-Y$, where the link $U-X$ is dictated by the statistical model, and the link $X-Y$ is subject to the communication constraint $I(X;Y)\le k$. While trivially $I(U;Y)\le I(U;X)$ and $I(U;Y)\le I(X;Y)$, neither of these two inequalities are typically sufficient to obtain a good lower bound. A strong data processing inequality (SDPI)
\begin{align}\label{eq.SDPI}
	I(U;Y) \le \gamma^*(U,X) I(X;Y), \qquad \forall p_{Y|X}
\end{align}
with $\gamma^*(U,X)<1$ can be desirable. The SDPI may take different forms (e.g., for $f$-divergences), and it is applied in most works on distributed estimation, e.g., \cite{zhang2013information,braverman2016communication,xu2017information}. The SDPI-based approach turns out to be tight in certain models (e.g., the Gaussian model \cite{zhang2013information,braverman2016communication}), while it is also subject to some drawbacks:
\begin{enumerate}
	\item The tight constant $\gamma^*(U,X)$ is hard to obtain in general;
	\item The linearity of \eqref{eq.SDPI} in $I(X;Y)$ can only give a linear dependence of $I(U;Y)$ on $k$, which may not be tight. For example, in Corollary \ref{cor.multi} the optimal dependence on $k$ is exponential;
	\item The conditional distribution $p_{Y^*|X}$ achieving the equality in \eqref{eq.SDPI} typically leads to $I(X;Y^*)\to 0$, and \eqref{eq.SDPI} may be loose for $I(X;Y)=k$;
	\item The operational meaning of \eqref{eq.SDPI} is not clear, which may not result in a valid encoding scheme from $X$ to $Y$.
\end{enumerate}

In contrast to the linear dependence on $k$ using SDPI, our technique implies that the dependence on $k$ is closely related to the tail of the score function. It would be an interesting future direction to explore other dependence on $k$ (instead of linear or exponential) in other statistical models.  

\section{Acknowledgements}
We are grateful to Jayadev Acharya, Cl\'ement Canonne, Himanshu Tyagi, and Rishabh Dudeja for spotting an error in handling interactive communication protocols in the earlier version \cite{han2018geometric}, and communicating with us. In particular, we would like to thank Jayadev Acharya, Cl\'ement Canonne, and Himanshu Tyagi for pointing out their recent paper \cite{acharya2020interactive} studying discrete distribution estimation models under the sequential communication protocol, which motivates the current proof technique to fix the error and handle the blackboard communication protocol for general models. 

\appendix
\section{Auxiliary Lemmas}
\begin{lemma}\cite{mitzenmacher2005probability}\label{lemma.poissontail}
	For $X\sim \mathsf{Poi}(\lambda)$ or $X\sim \mathsf{B}(n,\frac{\lambda}{n})$ and any $\delta>0$, we have
	\begin{align*}
		\mathbb{P}(X\ge (1+\delta)\lambda) &\le \left(\frac{e^\delta}{(1+\delta)^{1+\delta}}\right)^\lambda \le \exp\left(-\frac{(\delta^2\wedge \delta)\lambda}{3}\right),\\
		\mathbb{P}(X\le (1-\delta)\lambda) &\le \left(\frac{e^{-\delta}}{(1-\delta)^{1-\delta}}\right)^\lambda \le \exp\left(-\frac{\delta^2\lambda}{2}\right).
	\end{align*}
\end{lemma}

\begin{lemma}\cite{wyner1973theorem}\label{lemma.h_2}
	For the binary entropy function $h_2(x)\triangleq -x\log_2x-(1-x)\log_2(1-x)$ on $[0,\frac{1}{2}]$, let $h_2^{-1}(y)$ be its inverse for $y\in [0,1]$. Then the function
	\begin{align*}
		f(y)= (1-2h_2^{-1}(y))^2
	\end{align*}
	is a decreasing concave function, with $f(y)\le 2\log 2\cdot (1-y)$ for all $y\in [0,1]$.
\end{lemma}

\section{Proof of Main Lemmas}

\subsection{Proof of Lemma \ref{lemma.total_weight}}
We prove a stronger result: for any strategy $\{a_v(\cdot)\}$, if each path from the root to any leaf node visits exactly $k_i$ internal nodes with label $i$ for each $i\in [n]$, then
\begin{align}\label{eq.stronger}
	\sum_{y\in \{0,1\}^{\sum_{i=1}^n k_i}} \prod_{v\in \tau(y), l_v\neq i} b_{v,y}(x_{l_v}) = 2^{k_i}
\end{align}
for any $\{x_j\}_{j\neq i}$. Clearly \eqref{eq.stronger} implies the lemma (i.e., with $k_i=0$ and $k_i=k$, respectively). 

We prove \eqref{eq.stronger} by induction on the depth $D=\sum_{i=1}^n k_i$ of the binary tree. The base case $D=0$ is obvious. To move from $D$ to $D+1$, distinguish into two cases and apply the induction hypothesis to the left/right tree of the root: 
\begin{enumerate}
	\item If the root node is labeled as $i$, then \eqref{eq.stronger} follows from $2^{k_i}=2^{k_i-1}+2^{k_i-1}$; 
	\item If the root node is not labeled as $i$, then \eqref{eq.stronger} follows from $2^{k_i} = 2^{k_i}a_{\text{root}}(x_i) + 2^{k_i}(1-a_{\text{root}}(x_i))$. 
\end{enumerate}

\subsection{Proof of Lemma \ref{lemma.geometry_1}}
As $\|x\|_2 = \max_{u: \|u\|_2=1} u^\top x$, it suffices to prove the same upper bounds of $\bE[u^\top X\mid A]^2$ for any unit vector $u\in \bR^d$. First, by the Cauchy--Schwarz inequality, we have
\begin{align*}
	\bE[u^\top X\mid A]^2 \le \bE[(u^\top X)^2 \mid A] \le \frac{\bE[(u^\top X)^2]}{P(A)} \le \frac{I_0}{P(A)}, 
\end{align*}
establishing the first inequality. The improved inequality when $\bE[X] = 0$ is due to
\begin{align*}
	I_0 &\ge \bE[(u^\top X)^2] \\
	&=  \bE[(u^\top X)^2\mathbbm{1}(X\in A)] +  \bE[(u^\top X)^2\mathbbm{1}(X\in A^c)] \\
	&\stepa{\ge} \frac{\bE^2[(u^\top X)\mathbbm{1}(X\in A)]}{P(A)} + \frac{\bE^2[(u^\top X)\mathbbm{1}(X\in A^c)]}{1-P(A)} \\
	&\stepb{=} \frac{\bE^2[(u^\top X)\mathbbm{1}(X\in A)]}{P(A)} + \frac{\bE^2[(u^\top X)\mathbbm{1}(X\in A)]}{1-P(A)} \\
	&= \frac{P(A)}{1-P(A)}\cdot \bE[u^\top X\mid A]^2, 
\end{align*}
where (a) is due to Cauchy--Schwarz, and (b) follows from the assumption $\bE[X] = 0$. 

%

\subsection{Proof of Lemma \ref{lemma.geometry_2}}
By the definition of the Orlicz $\psi_2$-norm, for any unit vector $u\in \bR^d$ we have
\begin{align*}
	2 &\ge \bE\left[ \exp\left(\frac{(u^\top X)^2}{\Sigma_0} \right) \right] \\
	&\ge P(A)\cdot \bE\left[\exp\left(\frac{(u^\top X)^2}{\Sigma_0} \right) \bigg\vert A\right] \\
	&\ge P(A)\cdot \exp\left(\frac{(u^\top \bE[X\mid A])^2}{\Sigma_0}\right),
\end{align*}
where the last inequality follows from the convexity of $x\mapsto \exp(cx^2)$ for any $c>0$. Consequently, we have $u^\top \bE[X\mid A]\le \Sigma_0\log(2/P(A))$ for all unit vectors $u\in \bR^d$, and the result follows. 

\subsection{Another Proof of Lemma \ref{lemma.geometry_2} in Gaussian Case}
We prove the following lemma:
\begin{lemma}\label{lemma.geometry_gaussian}
	For $X\sim \calN(0,I_d)$ and any measurable $A\subseteq \bR^d$, we have
	\begin{align*}
		\left\|\bE[X\mid A]\right\|_2^2 \le 2\cdot \log \frac{1}{\bP(A)}.
	\end{align*}
\end{lemma}
We split the proof into two steps: we first consider the uniform distribution on the binary hypercube, and then use the tensor power trick to reduce to the Gaussian case.

\subsubsection{Geometric inequality on binary hypercube}
We prove the following lemma:
\begin{lemma}\label{lemma.hamming}
	For $X\sim \mathsf{Unif}(\{\pm 1\}^d)$ and any non-negative function $a(\cdot)\in [0,1]$, we have
	\begin{align*}
		\left\|\frac{\bE[Xa(X)]}{\bE[a(X)]}\right\|_2^2 \le 2\cdot \log \frac{1}{\bE[a(X)]}
	\end{align*}
	Moreover, the dimension-free constant $2$ cannot be improved.
\end{lemma}
\begin{proof}
	Define a new probability measure $Q(\cdot)$ on the binary hypercube $\{\pm 1\}^d$ with $Q(y)\propto a(y)$, and let $Y\sim Q$. Let $p_i\triangleq \bP(Y_i=1)$ for $i\in [d]$, then
	\begin{align*}
		\left\|\frac{\bE [Xa(X)]}{\bE [a(X)]}\right\|_2^2 = \|\bE Y\|_2^2 = \sum_{i=1}^d (\bE Y_i)^2 = \sum_{i=1}^d (1-2p_i)^2.
	\end{align*}
	Recall the definition of $h_2(\cdot)$ in Lemma \ref{lemma.h_2}. Define $q_i\triangleq h_2(p_i)$, the concavity in Lemma \ref{lemma.h_2} gives
	\begin{align*}
		\left\|\frac{\bE [Xa(X)]}{\bE [a(X)]}\right\|_2^2 = \sum_{i=1}^d (1-2h_2^{-1}(q_i))^2 \le d\left(1-2h_2^{-1}\left(\frac{1}{d}\sum_{i=1}^d q_i\right)\right)^2.
	\end{align*}
	
	On the other hand, by the subadditivity of Shannon entropy,
	\begin{align*}
		\sum_{i=1}^d q_i = \frac{1}{\log 2}\sum_{i=1}^d H(Y_i) &\ge \frac{H(Y)}{\log 2}= d-\bE\left[\log_2 \frac{a(Y)}{\bE[a(X)]}\right] \\
		&\ge d-\bE\left[\log_2 \frac{1}{\bE[a(X)]}\right] = d-\log_2 \frac{1}{\bE[a(X)]}.
	\end{align*}
	Hence, applying the decreasing property and the last inequality in Lemma \ref{lemma.h_2}, we have
	\begin{align*}
		\left\|\frac{\bE [Xa(X)]}{\bE [a(X)]}\right\|_2^2&\le d\left(1-2h_2^{-1}\left(1-\frac{1}{d}\log_2\frac{1}{\bE[a(X)]}\right)\right)^2 \\
		&\le d\cdot 2\log 2\cdot \frac{1}{d}\log_2\frac{1}{\bE[a(X)]}\\
		& = 2\log\frac{1}{\bE[a(X)]}.
	\end{align*}
	
	To show that $2$ is the best possible constant, pick $a(x)=\mathbbm{1}_B(x)$ where $B$ is the Hamming ball with center ${\bf 1}$ and radius $\epsilon d$. Direct computation gives the constant $2$ as $d\to\infty$ and $\epsilon\to 0$.
\end{proof}

\subsubsection{Tensor Power Trick}
Next we make use of Lemma \ref{lemma.hamming} to prove the Gaussian case. We apply the so-called \emph{tensor power trick}: we lift the dimension by making $B$ independent copies, and apply CLT to move to the Gaussian case as $B\to\infty$. This idea has been widely used in harmonic analysis and high-dimensional geometry, e.g., to prove the isoperimetric inequality for the Gaussian measure \cite{ledoux2005concentration}. 

Here the trick goes: fix any dimension $d$ and any function $a(\cdot)\in [0,1]$ defined on $\bR^d$. By a suitable approximation we may assume that $a(\cdot)$ is continuous. Now for any $B>0$, we define a new function $\tilde{a}(\cdot)$ on $\{\pm 1\}^{dB}$ as follows:
\begin{align*}
	\tilde{a}(X)=\tilde{a}(\{X_{i,j}\}_{i\in [d],j\in [B]}) \triangleq a\left(\frac{\sum_{j=1}^n X_{1,j}}{\sqrt{B}},\cdots,\frac{\sum_{j=1}^n X_{d,j}}{\sqrt{B}}\right).
\end{align*}
By symmetry, we have
\begin{align*}
	\|\bE [X\tilde{a}(X)]\|_2^2 = \sum_{i=1}^d \left(\bE\left[\frac{\sum_{j=1}^B X_{i,j}}{\sqrt{B}}a\left(\frac{\sum_{j=1}^n X_{1,j}}{\sqrt{B}},\cdots,\frac{\sum_{j=1}^n X_{d,j}}{\sqrt{B}}\right)\right]\right)^2.
\end{align*}
Moreover, by Lemma \ref{lemma.hamming}, we have
\begin{align}\label{eq.uniform}
	\left\|\frac{\bE [X\tilde{a}(X)]}{\bE [\tilde{a}(X)]}\right\|_2^2 \le 2\cdot \log \frac{1}{\bE[\tilde{a}(X)]}.
\end{align}

Let $Z\sim \calN(0,I_d)$, then CLT gives $\|\bE [X\tilde{a}(X)]\|_2^2 \to \|\bE [Za(Z)]\|_2^2$ and $\bE[\tilde{a}(X)]\to \bE[a(Z)]$ as $B\to\infty$. Hence, as $B\to\infty$, \eqref{eq.uniform} becomes
\begin{align}\label{eq.gaussian}
	\left\|\frac{\bE [Za(Z)]}{\bE [{a}(Z)]}\right\|_2^2 \le 2\cdot \log \frac{1}{\bE[a(Z)]}.
\end{align}
Note that \eqref{eq.gaussian} holds for all $d$ and $a(\cdot)$, the proof of Lemma \ref{lemma.geometry_gaussian} is complete by choosing $a(\cdot)=\mathbbm{1}_A(\cdot)$.

\section{Proof of Propositions}

\subsection{Proof of Proposition \ref{prop.assumption}}
\noindent\textbf{Product Bernoulli models.} We begin with the first product Bernoulli model. For any $\delta \in (0,1/6)$, we choose $\theta_u = (1/2,\cdots,1/2) + \delta u\in \Theta$ for all $u\in \{\pm 1\}^d$. Clearly the regular grid condition and the separation condition hold. For the likelihood ratio condition, note that 
\begin{align*}
	\frac{ dP_{\theta_{u^{\oplus j}}} }{dP_{\theta_u}}(x) = \frac{1-2\delta u_j}{1+2\delta u_j}\cdot \mathbbm{1}(x_j = 1) + \frac{1+2\delta u_j}{1-2\delta u_j}\cdot \mathbbm{1}(x_j = 0) \ge \frac{1}{2}
\end{align*}
by the choice of $\delta$. Moreover, the $j$-th component of the score function at $\theta_u$ is
\begin{align*}
	[S_{\theta_u}(x)]_j = \frac{2}{1+2\delta u_j}\cdot \mathbbm{1}(x_j = 1) - \frac{2}{1-2\delta u_j}\cdot \mathbbm{1}(x_j = 0), 
\end{align*}
therefore \eqref{eq:approximation} is satisfied with $\varepsilon \equiv 0$. 
\vspace{1em}

\noindent\textbf{Multinomial models.} For $d_0 = d/2$, consider the following construction known as the Paninski's construction \cite{paninski2008coincidence}: 
\begin{align*}
	\theta_u = \left(\frac{1}{d}-\frac{\delta u_1}{2}, \frac{1}{d}+\frac{\delta u_1}{2}, \cdots, \frac{1}{d}-\frac{\delta u_{d_0}}{2}, \frac{1}{d}+\frac{\delta u_{d_0}}{2}\right). 
\end{align*}
After proper permutation of the rows, it is easy to see that the matrix $M_u$ is $\delta\cdot [\text{diag}(v)\ \text{diag}(v)]^\top$ for some $v\in \{\pm 1\}^{d_0}$. Consequently, the operator norm of this matrix is $\sqrt{2}\delta$, which is smaller than $2\delta$. Also, after simple algebra, the separation condition \eqref{eq:separation} is fulfilled with $\kappa(\delta) = 2^{2-p}\delta^p$. For the likelihood ratio condition, note that
\begin{align*}
	\frac{ dP_{\theta_{u^{\oplus j}}} }{dP_{\theta_u}}(x) = \frac{2+d\delta u_j}{2-d\delta u_j}\cdot \mathbbm{1}(x = 2j-1) + \frac{2-d\delta u_j}{2+d\delta u_j}\cdot \mathbbm{1}(x = 2j) \ge \frac{1}{2}
\end{align*}
as $\delta \in (0,1/(2d))$. Moreover, although there is some ambiguity in defining the score function for the Multinomial model (depending on the choice of free parameters), the inner product $(\theta_{u^{\oplus j}} - \theta_u)^\top S_{\theta_u}(x)$ is well-defined and expressed as 
\begin{align*}
	(\theta_{u^{\oplus j}} - \theta_u)^\top S_{\theta_u}(x) = \frac{\delta u_j}{1/d-\delta u_j/2}\cdot \mathbbm{1}(x=2j-1) - \frac{\delta u_j}{1/d+\delta u_j/2}\cdot \mathbbm{1}(x=2j). 
\end{align*}
Therefore, \eqref{eq:approximation} holds with $\varepsilon\equiv0$. The product Bernoulli model is handled analogously. 
\vspace{1em}

\noindent\textbf{Gaussian location models.} Choose $\theta_u = \delta u\in \bR^d$ for all $u\in \{\pm 1\}^d$, then clearly the regular grid condition and the separation condition hold. Let $\calX_0 = \{x\in \bR^d: \|x\|_\infty \le (C\sqrt{\log(nd)}+1)\sigma \}$, then for a large enough constant $C>0$, using the Gaussian tail and the union bound gives that $P_{\theta_u}(\calX_0)\ge 1-o(n^{-1})$ for any $u\in \{\pm 1\}^d$ and $\delta < \sigma$. For the likelihood ratio condition, we have
\begin{align*}
	\frac{ dP_{\theta_{u^{\oplus j}}} }{dP_{\theta_u}}(x) = \exp\left(-\frac{2\delta u_j x_j}{\sigma^2} \right) \ge \frac{1}{2}, \quad \forall x\in \calX_0
\end{align*}
as $|\delta|\le c\sigma/\sqrt{\log(nd)}$ for a small enough constant $c>0$ and $|x_j|\le (C\sqrt{\log(nd)}+1)\sigma$. Moreover, $S_{\theta_u}(x) = (x-\theta_u)/\sigma^2$, and therefore 
\begin{align*}
	\left|\frac{ dP_{\theta_{u^{\oplus j}}} }{dP_{\theta_u}}(x) - 1 - (\theta_{u^{\oplus j}} - \theta_u)^\top S_{\theta_u}(x) \right| &= \left|\exp\left(-\frac{2\delta u_j x_j}{\sigma^2} \right) - 1 - \frac{2\delta u_j(\delta u_j- x_j)}{\sigma^2} \right| \\
	&= \left|\exp\left(-\frac{2\delta^2}{\sigma^2} - \frac{2\delta u_j z_j}{\sigma}\right) - 1 + \frac{2\delta u_j z_j}{\sigma} \right|, 
\end{align*}
with $x_j \triangleq \sigma z_j + \delta u_j$. Note that when $x\sim P_{\theta_u}$, we have $z_j\sim \calN(0,1)$, and therefore the above term has an explicit second moment as
\begin{align*}
	\bE_{X\sim P_{\theta_u}}\left[\left|\frac{ dP_{\theta_{u^{\oplus j}}} }{dP_{\theta_u}}(X) - 1 - (\theta_{u^{\oplus j}} - \theta_u)^\top S_{\theta_u}(X) \right|^2\right] = \exp\left(\frac{4\delta^2}{\sigma^2}\right) - 1 - \frac{4\delta^2}{\sigma^2},
\end{align*}
which is $\varepsilon^2$ with $\varepsilon = O(\delta^2/\sigma^2)$ as $\delta = O(\sigma)$. Hence, the Gaussian location model is approximately regular with $\varepsilon(\delta) = O(\delta^2/\sigma^2)$. 
\vspace{1em}

\noindent\textbf{Logistic regression models with random design.} Choose $\theta_u = \delta u\in \bR^d$ for all $u\in \{\pm 1\}^d$, then clearly the regular grid condition and the separation condition hold. For the likelihood ratio condition, we choose
\begin{align*}
	\calX_0 = \{(z,y): \|z\|_\infty \le C\sqrt{\log(nd)}, y\in \{0,1\} \}, 
\end{align*}
and for a large enough constant $C>0$ we have $P_{\theta}(\calX_0)\ge 1-\alpha$ with $\alpha= o(n^{-1})$ for any $\theta\in \bR^d$. We first show that for any fixed $y\in \{0,1\}$, taking only the expectation with respect to $z\sim \calN(0,I_d)$ satisfies \eqref{eq:approximation_weak}. By symmetry, we shall only consider the case $y=1$, where $X=(z,1)$ and 
\begin{align*}
	\left|\frac{dP_{\theta_{u^{\oplus j}}}}{dP_{\theta_u}}(z,1) - 1 - (\theta_{u^{\oplus j}} - \theta_u)^\top S_{\theta_u}(z,1)  \right| = \left|\frac{e^{-\theta^\top z}+1}{e^{-\theta^\top z}\cdot e^{2\delta u_jz_j}+1} - 1 + 2\delta u_jz_j\cdot \frac{e^{-\theta^\top z}}{e^{-\theta^\top z} + 1 }\right|. 
\end{align*}
We will prove that for any $A\ge 0$ and $t\in [-1/2,1/2]$, it holds that
\begin{align}\label{eq:algebra}
	\left|\frac{A+1}{Ae^t+1} - 1 + \frac{At}{A+1} \right| \le 2t^2. 
\end{align}
In fact, if \eqref{eq:algebra} holds, the choice of $\delta\in (0,1/\sqrt{d})$ satisfies $|2\delta u_jz_j|\le 1/2$ for any $X=(z,y)\in \calX_0$ with large $d$. Now choosing $A=e^{-\theta^\top z}$ and $t=2\delta u_jz_j\in [-1/2,1/2]$ in \eqref{eq:algebra} gives the desired inequality \eqref{eq:approximation_weak} with $\varepsilon=O(\delta^2)$. 

Next we prove the inequality \eqref{eq:algebra}. After simple algebra, it is equivalent to prove that
\begin{align*}
	\frac{A}{A+1}\cdot \left|\frac{[(t-1)e^t+1]A + (1+t-e^t)}{Ae^t+1}\right| \le 2t^2. 
\end{align*}
Since $|t|\le 1/2$, it is easy to verify that $|(t-1)e^t+1|\le t^2$ and $|1+t-e^t|\le t^2$, and clearly $e^t\ge 1/2$. Consequently, the above inequality holds, and so is \eqref{eq:algebra}. 

Finally, we verify $dP_{\theta_{u^{\oplus j}}}/dP_{\theta_u}(x)\ge 1/2$ for all $x\in \calX_0$. Using the above notations $A=e^{-\theta^\top z}$ and $t=2\delta u_jz_j\in [-1/2,1/2]$ again, this quantity is
\begin{align*}
	\frac{A+1}{Ae^t+1} \ge \min\left\{1, e^{-t} \right\} \ge \frac{1}{2},
\end{align*}
as desired.

\subsection{Proof of Proposition \ref{prop.geometry}}
Let $X$ follow the uniform distribution on $\Omega$, then $\bar{v}=\bE[X\mid A]$. As $X$ has independent coordinates each of which has a unit second moment, the assumption of Lemma \ref{lemma.geometry_1} is fulfilled with $I_0=1$. By Lemma \ref{lemma.geometry_1}, for $|A|=2^{d-1}$ we have
\begin{align*}
	\|\bE[X\mid A]\|_2 \le 1\cdot \frac{\bP(A)}{1-\bP(A)} = 1, 
\end{align*}
establishing the first inequality. Similarly, the second inequality follows from Lemma \ref{lemma.hamming} (and its proof). 

\section{Proof of Theorem \ref{thm.sparse}}
As the hypothesis class for sparse Gaussian models is typically not cube-like, we use the following Fano's inequality instead of the Assouad's lemma to establish the lower bound. The present form is taken from \cite[Corollary 1]{duchi2013distance}; see also \cite{chen2016bayes} and \cite[Theorem 8]{han2019online} for a general statement. 
\begin{lemma}\label{lemma.fano}
	Let random variables $V$ and $\widehat{V}$ take value in $\calV$, $V$ be uniform on some finite alphabet $\calV$, and $V-X-\widehat{V}$ form a Markov chain. Let $d$ be any metric on $\calV$, and for $t>0$, define
	\begin{align*}
		N_{\max}(t) &\triangleq \max_{v\in \calV} |v'\in V: d(v,v')\le t|, \\
		N_{\min}(t) &\triangleq \min_{v\in \calV} |v'\in V: d(v,v')\le t|.
	\end{align*}
	If $N_{\max}(t)+N_{\min}(t)<|\calV|$, the following inequality holds:
	\begin{align*}
		\mathbb{P}(d(V,\widehat{V})>t) \ge 1 - \frac{I(V;X)+\log 2}{\log\frac{|\calV|}{N_{\max}(t)}}.
	\end{align*}
\end{lemma}

We construct the following family of hypotheses: let $U\in\bR^d$ be uniformly distributed on the finite set
\begin{align*}
	\calU = \{\theta\in \{0,\pm 1\}^d : \|\theta\|_0=s\}.
\end{align*}
Clearly $|\calU|=2^s\binom{d}{s}$. For $u\in\calU$ we associate with the Gaussian distribution $P_u \triangleq \calN(\delta u,\sigma^2 I_d)$, draw $n$ i.i.d. observations $X=(X_1,\cdots,X_n)$ from $P_u$, and obtain the transcript $Y=(Y_1,\cdots,Y_n)\in \{0,1\}^{nk}$, where $Y_i\in \{0,1\}^k$ denotes the transcript from node $i$ under the simultaneous message passing protocol. Choosing $t=s/5$ in Lemma \ref{lemma.fano}, we have
\begin{align*}
	\left|\left\{u'\in\calU: d_{\mathsf{Ham}}(u,u')\le \frac{s}{5}\right\}\right| = \sum_{u+v\le \frac{s}{5}} \binom{s}{u}\binom{s-u}{v}\binom{d-s}{v} \le \left(\frac{s}{5}+1\right)^2\cdot \binom{s}{s/5}^2\binom{d}{s/5}.
\end{align*}
As a result, we have $\log \frac{|\calU|}{N_{\max}(s/5)}\ge cs\log\frac{d}{s}$ for some constant $c>0$, and Lemma \ref{lemma.fano} gives
\begin{align}\label{eq.sparse_lower_bound}
	\inf_{\Pi_{\mathsf{BB}}}\inf_{\widehat{\theta}}\sup_{\theta\in\Theta} \bE_\theta\|\widehat{\theta}-\theta\|_2^2 \ge \frac{s\delta^2}{10}\left(1-\frac{I(U;Y)+\log 2}{cs\log(d/s)}\right).
\end{align}

To lower bound \eqref{eq.sparse_lower_bound}, we seek an upper bound of $I(U;Y_i)$ for each $i\in [n]$. Under the simultaneous message passing protocol, the communication strategy of node $i$ could be represented by a family of non-negative functions $p_{i,y}(\cdot)$ with $y\in \{0,1\}^k$, where
\begin{align*}
	p_{i,y}(x) \triangleq \bP[Y_i = y \mid X_i = x]. 
\end{align*}
Clearly $\sum_{y\in \{0,1\}^k} p_{i,y}(x) = 1$ for all $x\in \bR^d$. Moreover, 
\begin{align*}
	\bP[Y_i=y \mid U = u] = \bE_{X_i\sim P_u}[p_{i,y}(X_i)]. 
\end{align*} 
Let $P_0\triangleq \calN(0,\sigma^2 I_d)$, we could upper bound the mutual information as
\begin{align*}
	I(U;Y_i) &\stepa{\le} \bE_U[D_{\mathsf{KL}}(P_{Y_i|U} \| P_{Y_i|U=0})] \\
	&\stepb{\le} \bE_U[\chi^2(P_{Y_i|U} \| P_{Y_i|U=0})] \\
	&\stepc{=} \sum_{y\in \{0,1\}^k}\bE_U\left[\frac{\bE_{X_i\sim P_0}^2[p_{i,y}(X_i)s_U(X_i)]}{\bE_{X_i\sim P_0}[p_{i,y}(X_i)]} \right] \\
	&\stepd{\le} \frac{2\delta^2}{\sigma^4} \sum_{y\in \{0,1\}^k} \bE_U\left[\frac{\bE_{X_i\sim P_0}^2[p_{i,y}(X_i)\cdot U^\top X_i]}{\bE_{X_i\sim P_0}[p_{i,y}(X_i)]} \right] + 2\sum_{y\in \{0,1\}^k}\bE_U\left[\frac{\bE_{X_i\sim P_0}^2[p_{i,y}(X_i)\varepsilon_U(X_i)]}{\bE_{X_i\sim P_0}[p_{i,y}(X_i)]} \right],
\end{align*}
where (a) is due to the variational representation of the mutual information $$I(X;Y) = \min_{Q_Y}\bE_X[D_{\mathsf{KL}}(P_{Y|X}\|Q_Y)],$$ (b) uses the fact that the KL divergence is upper bounded by the $\chi^2$ divergence, (c) follows from simple algebra with
\begin{align*}
	s_U(x) \triangleq \frac{dP_U}{dP_0}(x) - 1 = \exp\left(\frac{\delta\cdot U^\top x}{\sigma^2} - \frac{\delta^2 s}{2\sigma^2}\right) - 1, 
\end{align*}
and (d) uses the triangle inequality $(a+b)^2\le 2(a^2+b^2)$ with
\begin{align*}
	\varepsilon_U(x) \triangleq \exp\left(\frac{\delta\cdot U^\top x}{\sigma^2} - \frac{\delta^2 s}{2\sigma^2}\right) - 1 - \frac{\delta\cdot U^\top x}{\sigma^2}. 
\end{align*}

Next we upper bound each term separately. For the remainder term, in view of the identity
\begin{align*}
	\bE_{X\sim P_0}[\varepsilon_U(X)^2] = \exp\left(\frac{\delta^2 s}{\sigma^2} \right) - 1 - \frac{\delta^2 s}{\sigma^2} = O\left(\frac{\delta^4s^2}{\sigma^4}\right)
\end{align*}
as long as $\delta = O(\sigma/\sqrt{s})$, the Cauchy--Schwarz inequality gives
\begin{align}\label{eq:remainder}
	\sum_{y\in \{0,1\}^k}\bE_U\left[\frac{\bE_{X_i\sim P_0}^2[p_{i,y}(X_i)\varepsilon_U(X_i)]}{\bE_{X_i\sim P_0}[p_{i,y}(X_i)]} \right] \le \bE_{U}\bE_{X_i\sim P_0}\left[\sum_{y\in \{0,1\}^k} p_{i,y}(X_i)\varepsilon_U(X_i)^2 \right] = O\left(\frac{\delta^4s^2}{\sigma^4}\right).
\end{align}

As for the main term, we have
\begin{align}\label{eq:main_term}
	\sum_{y\in \{0,1\}^k} \bE_U\left[\frac{\bE_{X_i\sim P_0}^2[p_{i,y}(X_i)\cdot U^\top X_i]}{\bE_{X_i\sim P_0}[p_{i,y}(X_i)]} \right] &\stepa{=} \frac{s}{d}\sum_{y\in \{0,1\}^k} \frac{\| \bE_{X_i\sim P_0}[p_{i,y}(X_i)X_i] \|_2^2}{\bE_{X_i\sim P_0}[p_{i,y}(X_i)]} \nonumber \\
	&\stepb{\le} \frac{2s\sigma^2}{d}\sum_{y\in \{0,1\}^k} \bE_{X_i\sim P_0}[p_{i,y}(X_i)]\log \frac{1}{\bE_{X_i\sim P_0}[p_{i,y}(X_i)]} \nonumber\\
	&\stepc{\le} \frac{2s\sigma^2}{d}\cdot k,
\end{align}
where (a) follows from $\bE[UU^\top] = (s/d)I_d$, (b) follows from Lemma \ref{lemma.geometry_2} (or more precisely, Lemma \ref{lemma.geometry_gaussian}), and (c) uses the concavity of $x\mapsto \log x$. Now combining \eqref{eq:remainder} and \eqref{eq:main_term} gives the following upper bound on the mutual information: 
\begin{align*}
	I(U;Y_i) = O\left(\frac{sk\delta^2}{d\sigma^2} + \frac{s^2\delta^4}{\sigma^4} \right). 
\end{align*}
Without loss of generality we may assume that there is no public randomness (otherwise we use $I(U;Y)\le I(U;Y|R)$ for external randomness $R$ and repeat the previous arguments). Consequently, $(Y_1,\cdots,Y_n)$ are conditionally independent given $U$, and therefore
\begin{align}\label{eq:mutual_info}
	I(U;Y) \le \sum_{i=1}^n I(U;Y_i) =  O\left(\frac{nsk\delta^2}{d\sigma^2} + \frac{ns^2\delta^4}{\sigma^4} \right). 
\end{align}
Finally, choosing $\delta^2\asymp d\sigma^2\log(d/s)/(nk)$ in \eqref{eq.sparse_lower_bound} and \eqref{eq:mutual_info} completes the proof of Theorem \ref{thm.sparse} for $k\le d$ (also recall our choice of $n$). The case $k>d$ simply follows from the centralized minimax risk and is thus omitted. 

%

\bibliographystyle{alpha}
\bibliography{di,ref}

\end{document}